%% file: main.tex
\documentclass{article}
\usepackage[utf8]{inputenc}
\usepackage{amsmath, amsfonts, amsthm, amssymb}
\usepackage{enumitem}
\usepackage{graphicx}
\usepackage{booktabs}
\usepackage[letterpaper, vmargin=1in,hmargin=1.2in]{geometry}
\usepackage[ruled,vlined]{algorithm2e}
\usepackage{tikz}
\usetikzlibrary{cd}
\usepackage{pgfplots}\pgfplotsset{compat=1.16}
\usepackage{parskip} 
\usepackage[margin=40pt,font=small,labelfont=bf]{caption}
\usepackage{subcaption}
\usepackage{hyperref}
\usepackage[capitalise]{cleveref}

\newcommand{\RR}{\mathbb{R}}
\newcommand{\NN}{\mathbb{N}}
\newcommand{\B}{\mathcal{B}}
\newcommand{\D}{\mathcal{D}}
\renewcommand{\P}{\mathcal{P}}
\newcommand{\R}{\mathcal{R}}
\newcommand{\V}{\mathcal{V}}
\newcommand{\p}{\mathsf{p}}
\newcommand{\q}{\mathsf{q}}
\renewcommand{\r}{\mathsf{r}}
\newcommand{\prev}{\mathrm{prev}}

\theoremstyle{plain}
\newtheorem{theorem}{Theorem}[section]
\newtheorem{lemma}[theorem]{Lemma}

\theoremstyle{definition}
\newtheorem{definition}[theorem]{Definition}
\newtheorem{example}[theorem]{Example}

\SetCommentSty{mycommfont}

\title{\textbf{Value-Offset Bifiltrations for Digital Images}}
\author{Anway De, Thong Vo, Matthew Wright }
\date{\today}

\begin{document}

\maketitle

\begin{abstract}
  Persistent homology, an algebraic method for discerning structure in abstract data, relies on the construction of a sequence of nested topological spaces known as a filtration.
  Two-parameter persistent homology allows the analysis of data simultaneously filtered by two parameters, but requires a bifiltration---a sequence of topological spaces simultaneously indexed by two parameters.
  To apply two-parameter persistence to digital images, we first must consider bifiltrations constructed from digital images, which have scarcely been studied.
  We introduce the \emph{value-offset bifiltration} for grayscale digital image data.
  We present efficient algorithms for computing this bifiltration with respect to the taxicab distance and for approximating it with respect to the Euclidean distance.
  We analyze the runtime complexity of our algorithms, demonstrate the results on sample images, and contrast the bifiltrations obtained from real images with those obtained from random noise.

\end{abstract}

\section{Introduction}

Persistent homology can be used to discern structure in data, including digital image data.
This requires constructing from the digital image a filtration, which is a nested sequence of topological spaces.
A typical approach involves a \emph{sublevel set filtration} constructed from a grayscale digital image.
Unfortunately, single-parameter persistent homology computed from a sublevel set filtration suffers from several shortcomings.
It is susceptible to noise, as a single pixel with a value much different from its neighbors can appear as a long-lived feature in persistent homology.
Sublevel set persistence is also not well equipped to detect the size of sublevel sets or the distance between components of a sublevel set.

These shortcomings arise because a single-parameter filtration is not a rich enough structure to simultaneously capture both the range of grayscale intensity values and the distance information found in digital images.
Instead, digital images are a natural application for two-parameter persistent homology, which allows for the analysis of data simultaneously filtered by two parameters.
Two-parameter persistence requires the construction of a bifiltration, but bifiltrations from digital images have hardly been explored.
In this paper, we lay groundwork for two-parameter persistent homology of digital images by proposing and analyzing a bifiltration that captures both intensity and distance information from digital images.

Our approach is to augment the threshold process described above, expanding and contracting each sublevel set by a sequence of offset distances to create a two-parameter filtration we call the \emph{value-offset bifiltration}.
This construction is a special case of the sublevelset-offset filtration described in Michael Lesnick's thesis \cite{LesnickThesis}.
The value-offset bifiltration captures both grayscale intensity and distance information from the digital image, thus better representing the structure of the digital image than sublevel set persistence alone.

\paragraph{Our Contribution.}
First, we define the \emph{value-offset bifiltration} for grayscale digital image data. This bifiltration records both grayscale intensity values and distances between pixels.
Second, we present an algorithm for computing the value-offset bifiltration from a digital image with respect to the taxicab distance, as well as a variant of our algorithm which approximates the bifiltration with respect to the Euclidean distance.
Third, we apply our algorithms to sample images and analyze the results in terms of both runtimes and the sizes of the resulting bifiltrations.
We also give worst-case example images that maximize the size of the value-offset bifiltrations and the runtime of our algorithms.
Our results shed light on the computation and characteristics of value-offset bifiltrations, paving the way for their future use in two-parameter persistent homology.

\paragraph{Related Work.}
Various researchers have applied single-parameter persistent homology to digital images, but few have considered the two-parameter setting.
For example, Bendich et al.\ \cite{Bendich}, Chung and Day \cite{ChungDay}, Tymochko et al.\  \cite{Tymochko}, and many others have used sublevel set filtrations to compute single-parameter persistence of digital image data.

As previously mentioned, our value-offset bifiltration is a special case of the sublevelset-offset filtration described in Michael Lesnick's thesis \cite{LesnickThesis}, but our contribution is quite different. In his thesis, Lesnick focuses on algebraic stability and interleavings, while this paper presents and analyzes algorithms.

We describe our bifiltration in terms of thickening and thinning operations on sets of pixels, which are similar to the dilation and erosion operations from mathematical morphology. Chung, Day, and Hu use dilation and erosion to obtain multifiltrations and then explore applications to image denoising \cite{ChungDayHu}.
However, our construction relies on notions of distance between pixels, rather than structuring elements, to expand and contract sets of pixels.

Unknown to us until the conclusion of our work, Hu et al.\ recently defined a bifiltration on grayscale digital images involving the distance transform \cite{HuEtAl}. While the construction by Hu et al.\ is similar to ours, their focus is quite different: Hu et al.\ use their bifiltration for single-parameter persistence computations; we focus algorithms for constructing our bifiltration and analysis of its properties.

\paragraph{Outline.}
We formally define the value-offset bifiltration, along with related concepts for its construction, in \Cref{prelim}. We then present our algorithms for computing the value-offset bifiltration with respect to the taxicab distance in \Cref{taxicabSection}, followed by our approximate computation with respect to the Euclidean distance in \Cref{euclideanSection}. 
In \Cref{resultsSection} we present our experimental results, including runtimes and properties of our bifiltrations computed from sample images.
We conclude in \Cref{futureWork} with discussion and directions for future work.

\section{Mathematical Preliminaries}\label{prelim}

\subsection{Images and Filtrations}

We regard a digital image as a finite set of pixels $\P \subset \NN^2$ along with a function $f \colon \P \to \RR$ that assigns a grayscale intensity/color value to each pixel.
Let $N = |\P|$ be the number of pixels in the image.
In this paper, we assume that the intensity values are integers; that is, $f \colon \P \to \NN$.
In practice, these values are often chosen from the set $\{0, 1, 2, \ldots, 255\}$.
Let $\mathcal{V} = \mathrm{im}(f)$ denote the set of all intensity values in the image.

Though we define pixels to be points in $\NN^2$, we visualize pixels as small squares (as in \cref{exampleBifiltration} and \cref{imageExample}).
We use lowercase sans-serif letters (such as $\p$ and $\q$) to denote pixels; this will help avoid confusion later on.
When we occasionally need to refer to the coordinates of a pixel, we use $x$ and $y$ subscripts to indicate the horizontal and vertical coordinates, such as $\p = (\p_x, \p_y)$.
When illustrating $\p$ as a small square, we regard the square as centered at the point $(\p_x, \p_y)$.

A \emph{sublevel set} is the set of pixels whose values are less than or equal to a specified constant $v$:
\[ f^-_v = \{\p \in \mathcal{P} \mid f(\p) \le v\}. \]
A \emph{filtration} is a nested sequence of topological spaces. For any increasing sequence of constants $v_0 < v_1 < \cdots < v_k$, we obtain the sublevel set filtration
\[ f^-_{v_0} \subset f^-_{v_1} \subset f^-_{v_2} \subset \cdots \subset f^-_{v_k}. \]
For example, a sublevel set filtration appears in the green box within \cref{exampleBifiltration} for a $4\times 4$ image. 
Filtrations such as these often serve as the input for computing single-parameter persistent homology (e.g., as is done in \cite{Bendich, ChungDay, Tymochko}).

\subsection{Thickening and Thinning}

We denote the \emph{distance} between any two pixels $\p$ and $\q$ generically by $d(\p, \q)$.
In this paper, we consider the following two definitions of distance.

The \emph{taxicab distance} between two pixels is the sum of the absolute differences between their coordinates:
\begin{equation*}
    d_T(\p, \q) = |\p_x - \q_x| + |\p_y - \q_y| 
\end{equation*}

The \emph{Euclidean distance} is the usual straight-line distance between (the centers of) pixels:
\begin{equation*}
    d_E(\p, \q) = \sqrt{ (\p_x - \q_x)^2 + (\p_y - \q_y)^2} 
\end{equation*}

Given a notion of distance, we define ``thickening'' and ``thinning'' operations on any set of pixels.
Each of these operations requires an offset value that determines by how much the set of pixels is to be thickened or thinned. We adopt the convention that positive offset parameters specify thickening, while negative offset parameters specify thinning. In this work, we are primarily interested in thickening and thinning sublevel sets of digital images.

Let $\mathcal{S} \subset \mathbb{N}^2$ be a set of pixels and $t$ a positive offset value.
\emph{Thickening} $\mathcal{S}$ by $t$ creates a new set $\mathcal{S}_t \supset \mathcal{S}$ consisting of all pixels whose distance from $\mathcal{S}$ is less than or equal to $t$.
For a value $v$ and an offset $t > 0$, the sublevel set $f^-_v$ thickened by $t$ is the set
\[ f^-_{v,t} = \{\p \mid d(\p, \q) \le t \text{ for some } \q \text{ with } f(\q) \le v \}. \]
In other words, the thickened set $f^-_{v,t}$ consists of those pixels $\p$ that are distance $t$ or less from a pixel with value $v$ or less.

In contrast, given a set of pixels $\mathcal{S}$ and a negative offset value $t$, \emph{thinning} $\mathcal{S}$ by $-t$ removes all pixels that are distance $-t$ or less from a pixel not in $\mathcal{S}$, creating a new set $\mathcal{S}_t \subset \mathcal{S}$.
For a value $v$ and offset $t < 0$, the sublevel set $f^-_v$ thinned by $-t$ is the set
\[ f^-_{v,t} = \{\p \mid f(\q) \le v \text{ for all } \q \text{ with } d(\p,\q) \le -t \}. \]
In other words, when $t<0$, $f^-_{v,t}$ is the thinned set consisting of those pixels $\p$ such that all pixels at distance $-t$ or less have value $v$ or less.

Furthermore, we define $f^-_{v,0}$ to be simply the sublevel set $f^-_v$.
Thus, for any value $v \in \mathcal{V}$ and any $t \in \mathbb{R}$, the set of pixels $f^-_{v,t}$ is obtained from the sublevel set $f^-_v$ by either thickening if $t > 0$, thinning if $t < 0$, or doing nothing if $t=0$.
Importantly, thickening and thinning produce a filtration: for any increasing sequence $t_0 < t_1 < \cdots < t_k$, we obtain a filtration
\[ f^-_{v,t_0} \subset f^-_{v,t_1} \subset f^-_{v,1t_2} \subset \cdots \subset f^-_{v,t_k}. \]

For an example, \cref{exampleBifiltration} illustrates thickening and thinning with respect to the taxicab distance. For a given value $v$, the set of pixels shaded in the $t=0$ row of \cref{exampleBifiltration} is expanded to include neighboring pixels horizontally and vertically, producing the set of pixels shaded in the $t=1$ row. The sets of pixels shaded in the $t=2$ and $t=3$ rows are obtained similarly.
Likewise, \cref{exampleBifiltration} illustrates thinning for negative offset parameters:
for a given value of $v$, consider the set of pixels shaded in the $t=0$ row of the figure, then remove all pixels that have an unshaded neighbor (horizontally or vertically) to obtain the set of pixels shaded in the $t=-1$ row. Again remove all pixels with an unshaded neighbor to obtain the set of pixels shaded in the $t=-2$ row.
Each column (i.e., fixed value $v$) in \cref{exampleBifiltration} gives an example of a thickening/thinning filtration.

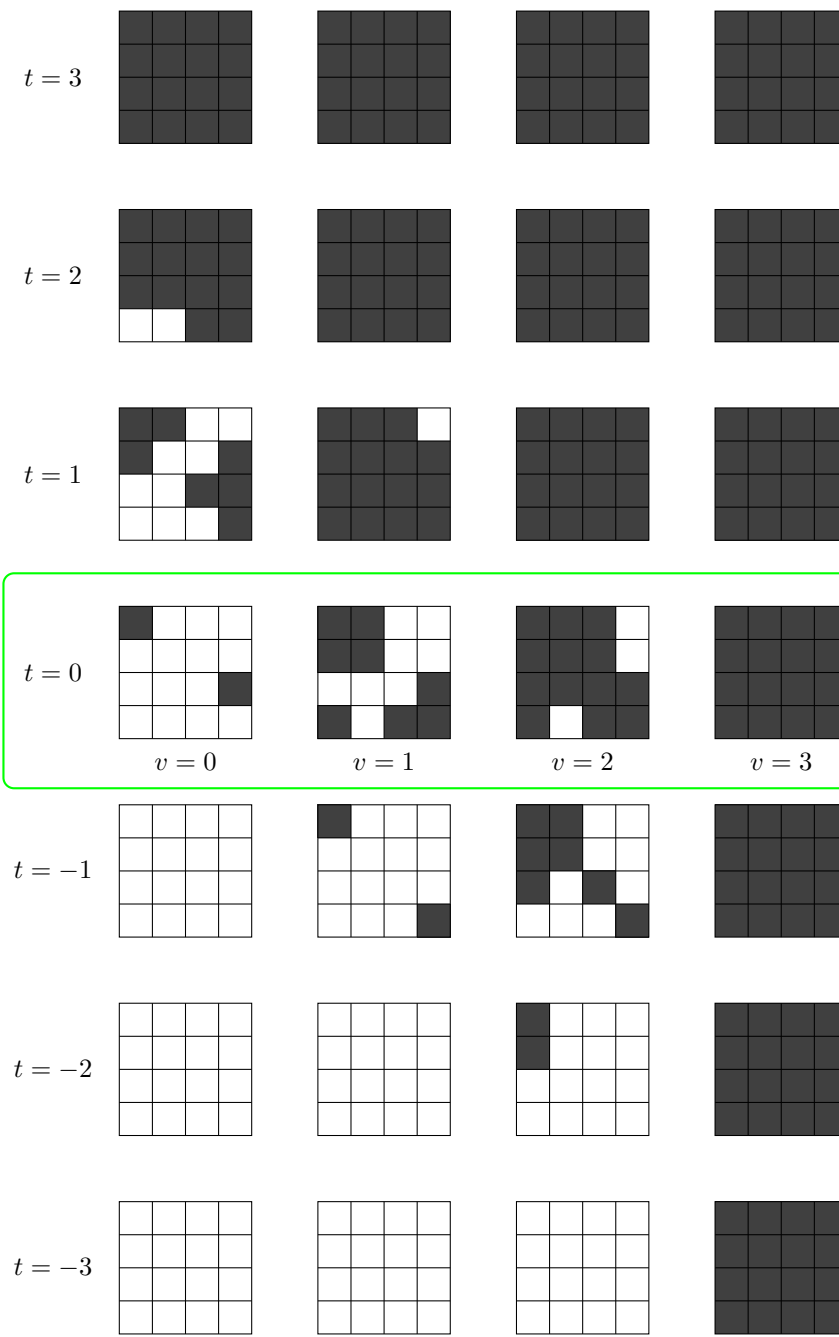
\begin{figure}[p]
  \centering
  \input{bifiltration_example}
  \caption{A $4\times 4$ image with values $\{0,1,2,3\}$ is shown at top, and its value-offset bifiltration with respect to the taxicab distance is shown below. The highlighted row $(t=0)$ shows the sublevel set filtration for the digital image.}
  \label{exampleBifiltration}
\end{figure}

Observe that thinning can be thought of as thickening the complement of a set of pixels: for example, thinning the shaded region in \cref{exampleBifiltration} is the same as thickening the unshaded region. We elaborate on this below (see especially \cref{thinningAndComplement}).

\subsection{The Value-Offset Bifiltration}

Applied to a sublevel set filtration, the thickening and thinning operations produce a \emph{bifiltration}: a collection of subsets of pixels indexed by two parameters, with inclusions in the direction of increase of each parameter.
Concretely, since $f^-_{v,t}$ is the sublevel set $f^-_v$ thickened (or thinned) by positive (or negative) $t$, the value-offset bifiltration is depicted as follows:

\begin{center}
  \begin{tikzcd}[row sep=scriptsize]
    & \vdots & \vdots & \vdots & \\
    \cdots \ar[r,hook] & f^-_{v_0,t_2} \ar[r,hook] \ar[u,hook] & f^-_{v_1,t_2} \ar[r,hook] \ar[u,hook] & f^-_{v_2,t_2} \arrow[r,hook] \ar[u,hook] & \cdots \\
    \cdots \ar[r,hook] & f^-_{v_0,t_1} \ar[r,hook] \ar[u,hook] & f^-_{v_1,t_1} \ar[r,hook] \ar[u,hook] & f^-_{v_2,t_1} \arrow[r,hook] \ar[u,hook] & \cdots \\
    \cdots \ar[r,hook] & f^-_{v_0,t_0} \ar[r,hook] \ar[u,hook] & f^-_{v_1,t_0} \ar[r,hook] \ar[u,hook] & f^-_{v_2,t_0} \arrow[r,hook] \ar[u,hook] & \cdots \\
    & \vdots \ar[u,hook] & \vdots \ar[u,hook] & \vdots \ar[u,hook] & 
  \end{tikzcd}
\end{center}

\cref{exampleBifiltration} shows a value-offset bifiltration with respect to the taxicab distance for a $4\times 4$ digital image.
Note that each row and column of the figure is itself a (single-parameter) filtration.
In the figure, note that any pixel $\p$ shaded at a $(v,t)$ pair is also shaded at all other $(v',t')$ with $v \le v'$ and $t \le t'$.

\begin{definition}
    The \emph{value-offset bifiltration} is the collection $\{f^-_{v,t}\}_{v,t}$.
    We refer to an index pair $(v,t)$ as a \emph{bigrade}.
\end{definition}
We denote bigrades using lowercase bold letters, such as \textbf{b}, to distinguish bigrades from pixels.

The value-offset bifiltration can be viewed through either a continuous or discrete perspective.
From the continuous perspective, a bigrade may be any pair of real numbers; i.e., $(v,t) \in \RR^2$.
However, since a digital image contains finitely many pixels, a value-offset bifiltration involves only finitely many distinct sets $f^-_{v,t}$.
It suffices to adopt a discrete perspective in which all bigrades $(v,t)$ satisfy $v \in \V$ and either $t \in \D$ or $-t \in \D$, where $\D$ is the set of all distances between pixels in the image (see, for example, \cref{exampleBifiltration}).
The discrete perspective is particularly important for computational purposes.
In this paper, we use the discrete perspective unless otherwise noted.

\begin{definition}\label{presentPos}
    We say that pixel $\p$ is \emph{present} at bigrade $\mathbf{b} = (v,t)$ if $\p \in f^-_{v,t}$.
    Let $\R^+_\p$ denote set of all bigrades with $t \ge 0$ at which $\p$ is present.
    Similarly, let $\R^-_\p$ denote the set of all bigrades with $t \le 0$ at which $\p$ is present.
    Let $\R_\p = \R^-_\p \cup \R^+_\p$ denote the set of all bigrades at which $\p$ is present.
\end{definition}

In words, $\p$ is present at $(v,0)$ if $\p$ has value $f(\p) \le v$.
Pixel $\p$ is present at $(v,t)$ with $t > 0$ if there exists a pixel $\q$ such that $f(\q) \le v$ and $d(\p, \q) \le t$. 
The set of bigrades $\R^+_\p$ results from the thickening process.
Likewise, $\p$ is present at $(v,t)$ with $t < 0$ if for all pixels $\q$ such that $d(\p, \q) \le -t$, $f(\q) \le v$.
The set of bigrades $\R^-_\p$ results from the thinning process.

\begin{example} 
    Let $\p$ be the pixel in the lower-left corner of the $6 \times 6$ image in \cref{imageExample}.
    \Cref{imageExample} (a) highlights all pixels at taxicab distance $5$ or less from $\p$. Since there is a pixel of value $0$ at taxicab distance $5$ from $\p$, pixel $\p$ is present at bigrade $(0,5)$ in the value-offset bifiltration with respect to the taxicab distance.
    \Cref{imageExample} (b) highlights all pixels at taxicab distance $3$ or less from $\p$. Since these highlighted pixels have no value greater than $5$, pixel $\p$ is present at bigrade $(5,-3)$ in the same bifiltration.
    Lastly, \Cref{imageExample} (c) highlights all pixels at Euclidean distance $\sqrt{13}$ or less from $\p$. Because these highlighted pixels include some with value $5$ but none with value greater than $5$, pixel $\p$ is present at bigrade $(5,-\sqrt{13})$ but not at bigrade $(4, -\sqrt{13})$ in the value-offset bifiltration with respect to the Euclidean distance.
\end{example}

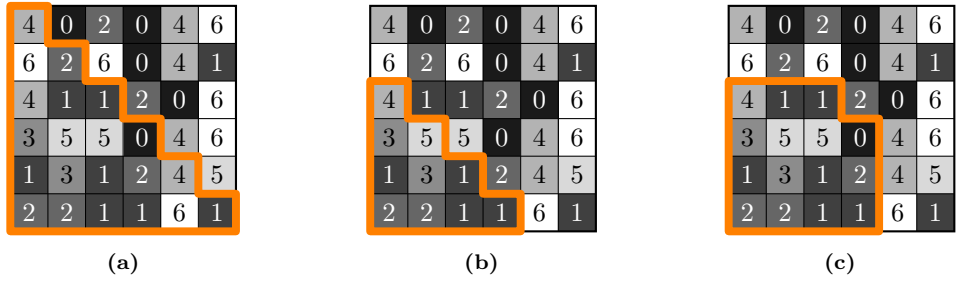
\begin{figure}[t]
    \input{imageExample}
    
    \caption{Let $\p$ be the lower-left pixel of this $6\times 6$ image, which is the same in all three frames. The region highlighted in (a) contains pixels of taxicab distance $5$ or less from $\p$. In (b), the highlighted region contains pixels of taxicab distance $3$ or less from $\p$. In (c), the highlighted region contains pixels of Euclidean distance $\sqrt{13}$ or less from $\p$.}
    \label{imageExample}
\end{figure}

Computing multiparameter persistent homology of the value-offset bifiltration requires first determining, for each pixel $\p$ in a given image, all pairs $(v,t)$ at which pixel $\p$ is present in the bifiltration; i.e., $\p \in f^-_{v,t}$.
However, this problem can be reduced to computing an ``entrance set'' for each pixel, as we explain next.

\subsection{Entrance Points and Entrance Sets}

We employ the \emph{partial order} $\preceq$ on bigrades given by
$\mathbf{a} = (v_\mathbf{a}, t_\mathbf{a}) \preceq (v_\mathbf{b}, t_\mathbf{b}) = \mathbf{b}$ if $v_\mathbf{a} \le v_\mathbf{b}$ and $t_\mathbf{a} \le t_\mathbf{b}$.
Furthermore, $\mathbf{a} \prec \mathbf{b}$ if $\mathbf{a} \preceq \mathbf{b}$ and $\mathbf{a} \ne \mathbf{b}$.
Two bigrades $\mathbf{a}$ and $\mathbf{b}$ are \emph{incomparable} if $\mathbf{a} \not\preceq \mathbf{b}$ and $\mathbf{b} \not\preceq \mathbf{a}$.
Bigrades may also be compared by a single coordinate; for this we write $\mathbf{a} <_v \mathbf{b}$ if $v_\mathbf{a} < v_\mathbf{b}$ or $\mathbf{a} <_t \mathbf{b}$ if $t_\mathbf{a} < t_\mathbf{b}$.

If a pixel $\p$ is present at bigrade $\mathbf{a}$, then $\p$ is also present at any bigrade $\mathbf{b}$ with $\mathbf{a} \preceq \mathbf{b}$.
This observation implies that $\R_\p$ has a ``stairstep'' shape, as illustrated in \cref{StairstepExample1}, and can be described by a set of minimal points in the following sense.

\begin{definition}
    A bigrade $\textbf{m}$ is \emph{minimal} in a set of bigrades $\mathcal{S}$ if there is no bigrade $\textbf{s} \in \mathcal{S}$ such that $\textbf{s} \prec \textbf{m}$.
    Similarly, a bigrade $\textbf{m}$ is a \emph{maximal} in a set of bigrades $\mathcal{S}$ if there does not exist $\textbf{s} \in \mathcal{S}$ such that $\textbf{m} \prec \textbf{s}$.
\end{definition}

Note that a set $\mathcal{S}$ may have many minimal (or maximal) bigrades, all mutually incomparable.

We now define several sets of ``entrance points'' that our algorithms will compute for each pixel.

\begin{definition}\label{entranceDef}
    The set of \emph{positive entrance points} $\B^+_\p$ associated with pixel $\p$ is the set of minimal bigrades in $\R^+_\p$. 
    The set of \emph{negative entrance points} $\B^-_\p$ associated with pixel $\p$ is the set of minimal bigrades in $\R^-_\p$ along with $(\max(\V),-\infty)$. 
    The set of \emph{entrance points}, or \emph{entrance set}, $\B_\p$ associated with pixel $\p$ is the set of minimal bigrades in $\B_\p^+ \cup \B_\p^-$.
\end{definition}

Intuitively, bigrade $(v,t)$ is an entrance point for pixel $\p$ if $\p$ is present at $(v,t)$ but not present at $(v',t)$ with $v'<v$ or at $(v,t')$ with $t'<t$.
The set of positive entrance points $\B^+_\p$ uniquely determines the bigrades at which $\p$ is present as a result of the thickening process; we note that $(f(\p),0)$ is always a positive entrance point for pixel $\p$.
Similarly, the set of negative entrance points uniquely determines the bigrades at which $\p$ is present as a result of the thinning process.
We define $(\max(\V),-\infty)$ to be a negative entrance point since pixel $\p$ is present at $(\max(\V),t)$ for all values $t$.
\cref{StairstepExample1} illustrates the entrance points for the lower-left pixel in \cref{imageExample}.

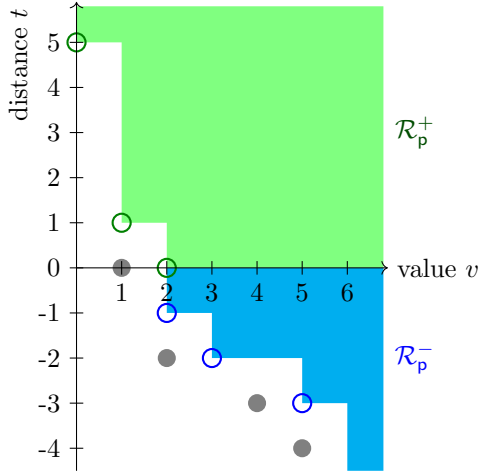
\begin{figure}[h]
    \centering
        \begin{tikzpicture}[scale=0.6]
            
            \fill[green!50!white] (0,5.8) -- (0,5) -- (1,5) -- (1,1) -- (2,1) -- (2,0) -- (6.8,0) -- (6.8,5.8) -- cycle;\
            \node[green!30!black] at (7.5,3) {$\R^+_\p$};
            \fill[cyan] (2,0)-- (2,-1) -- (3,-1) -- (3,-2) -- (5,-2) -- (5,-3) -- (6,-3) -- (6,-4.5) -- (6.8,-4.5) -- (6.8,0) -- cycle;
            \node[blue] at (7.5,-2) {$\R^-_\p$};
           
            \foreach \p in {(0,5),(1,1),(2,0)} {
                \draw[green!50!black,thick] \p circle (.2);
            }
            
            \foreach \p in {(2,-1),(3,-2),(5,-3)} {
                \draw[blue,thick] \p circle (.2);
            }
            
            \foreach \p in {(1,0),(2,-2),(4,-3),(5,-4)} {
                \fill[gray] \p circle (.2);
            }
    
            \draw[->] (0,-4.5) -- (0,5.9);
            \node[rotate=90] at (-1.3,4.5) {distance $t$};
            \draw[->] (0,0) -- (6.9,0) node[right] {value $v$};
            \foreach \i in {-4,-3,...,5} {
              \draw (-0.15,\i) node[left] {\i} -- (0.15,\i);
            }
            \foreach \i in {1,2,...,6} {
              \draw (\i,-0.15) node[below] {\i} -- (\i,0.15);
            }
        \end{tikzpicture}
        
    \caption{This plot shows the bigrades where pixel $\p$, the lower-left pixel in \cref{imageExample}, is present in the value-offset bifiltration with respect to the taxicab distance. The green shading depicts $\R^+_\p$ as a continuous region, with $\B^+_\p$ indicated by the green circled points. The blue shading depicts $\R^-_\p$ as a continuous region, with $\B^-_\p$ shown as blue circled points, though negative entrance point $(6, -\infty)$ is not shown. Note that bigrade $(2,0)$ is not in $\B_\p$ because it is not a minimal point of $\B^+_\p \cup \B^-_\p$. The set $\B^c_\p$ of complement entrance points is shown as gray dots.}
    \label{StairstepExample1}
\end{figure}

In the following sections, we present our algorithms for finding entrance points.
While we use a direct approach to find positive points, we employ a sort of duality to find negative entrance points.
Recall that our thinning procedure can be viewed as thickening the complement: precisely, thinning the sublevel set $f^-_v$ is equivalent to thickening the superlevel set $f^+_v = \{\p \in \P \mid f(\p) > v\}$.
This suggests that the same algorithm may be employed to find minimal bigrades at which a pixel $\p$ is present and also maximal bigrades at which $\p$ is \emph{not} present. 
The set of maximal bigrades $(v,t)$ with $t \le 0$ at which $\p$ is not present then determines $\B^-_\p$ in a simple way.
This is our approach, which we make precise by the following definition and theorem.

\begin{definition}\label{compDef}
    Let $\R_\p^c = \{ (v,t) \mid v \in \V, -t \in \D, \p \not\in f^-_{v,t} \}$ denote the set of bigrades with offset $t \le 0$ at which pixel $\p$ is \emph{not present}.
    The set of \emph{complement entrance points} $\B_\p^c$ is the set of maximal points in $\R^c_\p$.
\end{definition}

Intuitively, bigrade $(v,t)$ with $t < 0$ is a complement entrance point for $\p$ if all pixels with distance less than $-t$ from $\p$ have value less than or equal to $v$, and there is a pixel at distance $-t$ from $\p$ with value $v+1$.
Equivalently, $(v,t)$ with $v \in \V$, $t < 0$, $-t \in \D$ is a complement entrance point for pixel $\p$ if $\p$ is not present at $(v,t)$ but $\p$ is present at both $(v',t)$ for $v' > v$ and $v' \in \V$, and also at $(v,t')$ for $t' > t$ and $-t' \in \D$.

Note that if $f(\p) > \min(\V)$, then $(f(\p)-1, 0)$ is a complement entrance point, since it is a maximal point among the discrete bigrades $(v,t)$ with $t \le 0$ at which $\p$ is not present.
However, if $f(\p) = \min(\V)$, then there is no complement entrance point $(v,t)$ with $t=0$.

The following theorem establishes a relationship between negative entrance points and complement entrance points that is important for our algorithms.
The theorem requires a bit of notation: given any pixel $\p$ and distance function $d$, let $\D_\p =  \{ d(\p, \q) \mid \q \in \P \}$ be the set of distances from $\p$ to all pixels.
Since the image contains finitely many pixels, $\D_\p$ is a finite set.
For any $r>0$, let $\prev(r, \D_\p)$ be the largest value in $\D_\p$ that is smaller than $r$.
Furthermore, we say that complement entrance points $\mathbf{a}$ and $\mathbf{b}$ for pixel $\p$ are \emph{consecutive} if $\mathbf{a} <_v \mathbf{b}$ and there is no $\mathbf{c} \in \B^c_\p$ such that $\mathbf{a} <_v \mathbf{c} <_v \mathbf{b}$. We similarly refer to consecutive negative entrance points.

\begin{theorem} \label{thinningAndComplement}
For any two consecutive complement entrance points $\textbf{a} = (v_\textbf{a}, t_\textbf{a}) <_v (v_\textbf{b}, t_\textbf{b}) = \textbf{b}$ for a pixel $\p$, there exists a unique negative entrance point $\textbf{c}$ for $\p$ which satisfies $\textbf{a} <_v  \textbf{c} \le_v \textbf{b}$. 
\end{theorem}
\begin{proof}
    Let $t = -\prev(-t_\textbf{b}, \D_\p)$ and $\textbf{c} = (v_\textbf{a} + 1, t)$, as shown in \cref{consecutiveComplement}.
    We claim that $\textbf{c}$ is the unique negative entrance point guaranteed by the theorem.
    
    First, we show that $\p$ is present at bigrade $\mathbf{c}$.
    If $v_\textbf{b} = v_\textbf{a} + 1$, then $\mathbf{c} = (v_\textbf{b}, t)$ and $\p$ is present at $\mathbf{c}$ since $\mathbf{b}$ is a complement entrance point.
    Otherwise, $v_\textbf{b} > v_\textbf{a} + 1$.
    Since $\mathbf{a}$ is a complement entrance point, it follows that $\p$ is present at bigrade $(v_\textbf{a} + 1, t_\textbf{a})$.
    If $\p$ were not present at $\mathbf{c}$, then there would have to be a complement entrance point $(v_\textbf{a} + 1, t^*)$ with $t < t^* < t_\textbf{a}$, which would contradict the assumption that $\mathbf{a}$ and $\mathbf{b}$ are consecutive.
    Thus, $\p$ is present at $\mathbf{c}$. 
    
    Complement entrance points $\mathbf{a}$ and $\mathbf{b}$ guarantee that $\p$ is not present at bigrades $(v_\textbf{a},t)$ or $(v_\textbf{a} + 1, t_\textbf{b})$, respectively. Hence, $\mathbf{c}$ is a minimal point in $\R^-_\p$, meaning that $\mathbf{c}$ is a negative entrance point for pixel $\p$.
    
    Lastly, $\p$ cannot have any other negative entrance points $\mathbf{e} = (v_\textbf{e}, t_\textbf{e})$ with $v_\textbf{a} < v_\textbf{e} < v_\textbf{b}$. 
    If $t_\textbf{c} \le t_\textbf{e}$, then $\mathbf{c} \preceq \mathbf{e}$, so $\mathbf{e}$ is not minimal.
    However, if $t_\textbf{e} < t_\textbf{c}$, then $\p$ is not present at $\mathbf{e}$.
    
    Therefore, $\mathbf{c}$ is the unique negative entrance point satisfying $\textbf{a} <_v  \textbf{c} <_v \textbf{b}$.
\end{proof}

\begin{figure}[h]
    \centering
    \begin{tikzpicture}[scale=0.6]
        \draw (3,0.15) -- (3,-0.15);
        \draw (-0.15,-1) -- (0.15,-1);
        \draw[dashed] (3,0) -- (3,-6);
        \draw[dashed] (0,-1) -- (10,-1);
        \node at (3.4,-1.4) {$\mathbf{a}$};
        \node at (-0.5,-1) {$t_\mathbf{a}$};
        \node at (3,0.5) {$v_\mathbf{a}$};
        
        \draw (8,0.15) -- (8,-0.15);
        \draw (-0.15,-5) -- (0.15,-5);
        \draw[dashed] (8,0) -- (8,-6);
        \draw[dashed] (0,-5) -- (10,-5);
        \node at (8.4,-5.4) {$\mathbf{b}$};
        \node at (-0.5,-5) {$t_\mathbf{b}$};
        \node at (8,0.5) {$v_\mathbf{b}$};
        
        \draw (4,0.15) -- (4,-0.15);
        \draw (-0.15,-4) -- (0.15,-4);
        \draw[dashed] (4,0) -- (4,-6);
        \draw[dashed] (0,-4) -- (10,-4);
        \node at (4.4,-4.4) {$\mathbf{c}$};
        \node[left] at (-0.1,-4) {$t = -\prev(-t_\mathbf{b}, \D_\p)$};
        \node at (4,1.5) {$v_\mathbf{a}+1$};
        \draw[->] (4,1.1) -- (4,0.3);
        
        \foreach \p in {(3,-1),(8,-5)} {
            \fill[gray] \p circle (.2);
        }
            
        \fill[blue] (4,-4) circle (.2);
        
        \draw[<->] (0,-7) node[left] {offset $t$} -- (0,1);
        \draw[->] (-0.2,0) node[left] {$0$} -- (10,0) node[right] {value $v$};
        \node at (6,0.5) {$\cdots$};
        \node at (-0.5,-2.5) {$\vdots$};
    \end{tikzpicture}
    \caption{Between any two consecutive complement entrance points $\mathbf{a}$ and $\mathbf{b}$ there exists a unique negative entrance point $\mathbf{c}$, as shown in \cref{thinningAndComplement}.}
    \label{consecutiveComplement}
\end{figure}
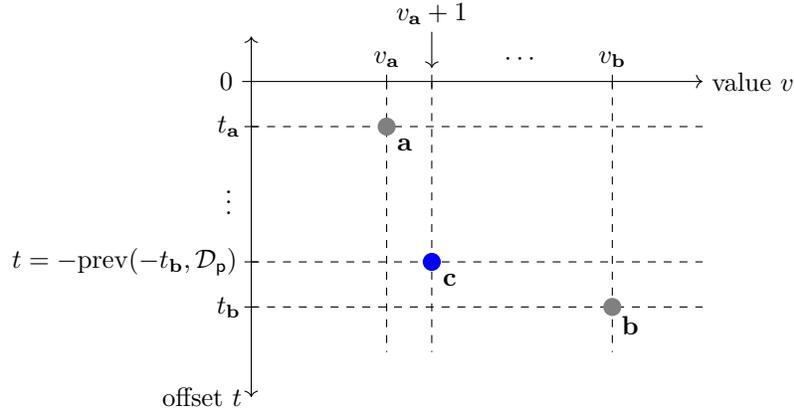

Similarly, between any two consecutive negative entrance points $\textbf{a} <_v \textbf{b}$, there exists a complement entrance point $\textbf{c}$ that satisfies $\textbf{a} <_v  \textbf{c} <_v \textbf{b}$ and $\textbf{b} <_t  \textbf{c} <_t \textbf{a}$; the proof is analogous to that of the previous theorem.
Therefore, if we find all of the complement entrance points for a pixel $\p$, we can apply the reasoning in the proof of \cref{thinningAndComplement} to find all negative entrance points, except for the last and possibly the first negative entrance point for $\p$.
In particular, the negative entrance point $(\max(\V),-\infty)$ does not lie between two complement entrance points.

\section{Taxicab Distance}\label{taxicabSection}

In this section we discuss our algorithm for finding the entrance set associated with each pixel in with respect to the taxicab distance.
We present our algorithm, prove its correctness, and discuss its time and space complexity.

\subsection{Algorithm}

We handle the thickening and thinning processes separately. 
We first find the positive entrance points (which result from thickening) and then find the negative entrance points (which result from thinning).
For both processes, and for each function value $v$, we perform a breadth-first search starting with all pixels of value $v$.
For each pixel, we maintain lists of all positive and negative entrance points.
We join these lists to obtain the entrance set for each pixel.

The core of our algorithm is the breadth-first search, allowing us to visit all pixels in order of increasing taxicab distance from the closest pixel with a given value $v$.
For this, we view the image as a graph, with each pixel a node and edges connecting pixels that are adjacent horizontally or vertically. 
For each value $v \in \mathcal{V}$, we employ a breadth-first traversal of this graph, starting at all pixels with value $v$.
The depth $T_{\p,v}$ at which pixel $\p$ is encountered in the breadth-first traversal at value $v$ gives the taxicab distance from pixel $\p$ to the nearest pixel with value $v$.

Our thickening algorithm, \cref{Thick4C}, iterates through the function values $v\in \mathcal{V}$ in ascending order.
For each pixel $\p$ with a function value $v$, we set $T_{\p,v}=0$ and push $\p$ into a queue. 
The queue is key to our breadth-first search.
Specifically, when a pixel $\p$ is removed from the queue, we consider all pixels $\q$ that are adjacent to $\p$ either horizontally or vertically.
For each such $\q$, we set $T_{\q,v} = T_{\p,v} + 1$ and check whether $(v,T_{\q,v})$ is a positive entrance point---that is, whether it is a new minimal point in $\mathsf{B}_\q^+$, our working positive entrance set for $\q$.
If so, we append $(v,T_{\q,v})$ to $\mathsf{B}_\q^+$ and push $\q$ into the queue.
We continue this process until the queue is empty, and then repeat for the next value of $v$.
Since each pixel has at most one positive entrance point per value $v$, each pixel is added to the queue at most once for each value $v$, and thus the algorithm terminates. 

\begin{algorithm}[h]
\ForEach{$v\in\mathcal{V}$ in ascending order}{
    \tcp{prepare breadth-first search at value $v$}
    
    initialize empty queue $Q$
    
    \ForEach{pixel $\p$ such that $f(\p) = v$}{ 
        $T_{\p,v} \gets 0$
        
        append bigrade $(v,0)$ to $\mathsf{B}^+_\p$
        
        $Q.\mathrm{push}(\p)$
    }
    
    \tcp{perform breadth-first search}
    \While{$Q$ is not empty}{
    
        $\p \gets Q.\mathrm{pop}()$
        
        \ForEach{pixel $\q$ adjacent to $\p$}{
            $T_{\q,v} \gets T_{\p,v}+1$
            
            \If{$\mathsf{B}^+_\q=\emptyset$ or $T_{\q,v} < \min\{y \mid (x,y)\in \mathsf{B}^+_\q\}$ }{
                append bigrade $(v,T_{\q,v})$ to $\mathsf{B}^+_\q$
                
                $Q.\mathrm{push}(\q)$
            }
        }
    }
}
\caption{Thickening with respect to Taxicab Distance}
\label{Thick4C}
\end{algorithm}

We offer a few comments about \cref{Thick4C}.
First, in the iteration for function value $v$, pixels with value less than $v$ are not added to the queue.
For, if pixel $\q$ has value $v' < v$, then the algorithm has already found an entrance point $(v', 0)$ for $\q$, so no bigrade with value $v$ can be an entrance point for $\q$.
Thus, $\q$ is never pushed into the queue during the breadth-first traversal for function value $v$.

Second, it is only necessary to push pixel $\q$ into the queue if its working entrance set is modified. Pushing $\q$ into the queue simply guarantees that the algorithm will later examine neighbors of $\q$ for entrance points at value $v$.
If $\mathsf{B}^+_\q$ is not modified at value $v$, that means it already contains entrance point $\mathbf{a} \prec (v, T_{\q,v})$; in this case, if any neighbor $\q'$ of $\q$ has an entrance point at value $v$, there must be a path from $\q'$ to a pixel of value $v$ that is shorter than any path through $\q$.

As a modification to the algorithm, whenever a entrance point is appended to $\mathsf{B}^+_\q$, we could mark pixel $\q$ as ``visited'' and specify that the innermost foreach loop considers only \emph{unvisited} pixels $\q$. This modification would not affect the output of the algorithm, but it would require additional memory for minimal speedup.

Our thinning algorithm, \cref{Thin4C}, operates similarly, except that it finds the complement entrance set for each pixel.
This involves iterating through the function values in decreasing order and finding maximal bigrades where each pixel is \emph{not} present in the bifiltration.

We start at $v = \max(\V)$ and iterate through the function values in descending order. 
As in \cref{Thick4C}, we employ a queue to perform a breadth-first traversal.
For each pixel $\p$ of value $v$, we find all pixels $\q$ such that $\p$ (at distance $t$) is a closest pixel of value $v$ or more. Such a pixel $\q$ is present at bigrades $(v,-t)$ and $(v-1, -\prev(-t))$, but not at $(v-1, -t)$;
thus, bigrade $(v-1,-t)$ is a complement entrance point for $\q$.
The algorithm compiles a working complement entrance set $\mathsf{B}_\p^c$, for each pixel.
At the end, we have found all the complement entrance points of the pixels; that is, $\mathsf{B}^c_\p=\mathcal{B}^c_\p$ for all pixels $\p$. 

We do not have to consider $v = \min(\V)$ since no pixels are present at bigrades with value $\min(\V)-1$.
However, when $v > \min(\V)$, we make sure that $(v-1,0) \in \mathsf{B}^c_\p$, since this is a complement entrance point by \cref{compDef}.
Note that in \Cref{Thin4C}, the values $T_{\q,v}$ are positive distance values, which are negated when used as offset values in bigrades.

\begin{algorithm}[h]
\ForEach { $v\in\mathcal{V}$ such that $v > \min(\V)$, in descending order,}{
    \tcp{prepare breadth-first search at value $v$}
    
    initialize empty queue $Q$
    
    \ForEach{pixel $\p$ such that $f(\p) = v$}{
        $T_{\p,v} \gets 0$
        
        append bigrade $(v-1, 0)$ to $\mathsf{B}^c_\p$
        
        $Q.\mathrm{push}(\p)$
    }
    \tcp{perform breadth-first search}
    \While{$Q$ is not empty}{

        $\p \gets Q.\mathrm{pop}()$
        
        \ForEach{pixel $\q$ adjacent to $\p$}{
            $T_{\q,v} \gets T_{\p,v} + 1$
            
            \If{$\mathsf{B}^c_\q=\emptyset$ or $T_{\q,v} < \min\{ -y \mid (x,y)\in \mathsf{B}^c_\q \}$}{
                
                append bigrade $(v-1, -T_{\q,v})$ to $\mathsf{B}^c_\q$
                
                $Q.\mathrm{push}(\q)$
            }
        }
    }
} 
\caption{Thinning with respect to Taxicab Distance}
\label{Thin4C}
\end{algorithm}

We find the negative entrance points from the complement entrance points for each pixel using \cref{thinningAndComplement}.
\Cref{ConversionAlg}, which we call our \emph{conversion} algorithm, not only performs this calculation but also combines the sets of positive and negative entrance points to produce the entrance set for each pixel.
Importantly, \cref{ConversionAlg} is independent of the notion of distance.
Our conversion algorithm proceeds one pixel at a time, iterating through each pair of consecutive complement entrance points in ascending order of function value.
The basic idea is simple: for each pair of consecutive complement entrance points, \cref{thinningAndComplement} gives a negative entrance point $\mathbf{b}$ for $\p$. 

However, the first negative entrance point must be handled carefully.
Specifically, a pixel $\p$ with value $f(\p) = \min(\V)$ has a negative entrance point with this value. This entrance point does not lie between two complement entrance points because no complement entrance points have values less than $\min(\V)$. If this entrance point is $(\min(\V),0)$, then it exists in $\mathsf{B}^+_\p$ from \cref{Thick4C}; otherwise, it is $(\min(\V),-t)$ for the largest $t \in \D_\p$ that is less than absolute offset value of all other negative entrance points for $\p$.

Furthermore, if pixel $\p$ has distinct positive and negative entrance points with value $f(\p)$, as in \cref{StairstepExample1}, then the positive entrance point $(f(\p),0)$ is not minimal in $\B_\p$.
Specifically, if pixel $\p$ has a negative entrance point at $(f(\p),t)$ for $t < 0$, then \cref{ConversionAlg} removes bigrade $(f(\p),0)$ from $\mathsf{B}^+_\p$.
This occurs in the first iteration of the inner foreach loop in \cref{ConversionAlg}, before the negative entrance bigrades are appended to $\mathsf{B}^+_\p$.

Lastly, we include an entrance point $(\max(\V), -\infty)$ for each pixel (recall \cref{entranceDef}).

At the end of this process, $\mathsf{B}_\p^+$ is the entrance set (that is, $\mathsf{B}_\p^+=\mathcal{B}_\p$) for each pixel $\p$, a claim we justify in the next subsection.

\begin{algorithm}[h]
  \ForEach{pixel $\p$}{

    \If{$f(\p) = \min(\V)$}{
        $t \gets -\prev( \min\{ -y \mid (x,y) \in \mathsf{B}^c_\p \}, \D_\p)$
        
        \If{$t < 0$}{
            remove bigrade $(\min(\V),0)$ from $\mathsf{B}^+_\p$
            
            append bigrade $(\min(\V), t)$ to $\mathsf{B}^+_\p$
        }
    }

    \ForEach{consecutive $\mathbf{a},\mathbf{b} \in \mathsf{B}^c_\p$ in ascending order of function value}{
        $v \gets v_\mathbf{a} + 1$
        
        $t \gets -\prev(-t_\mathbf{b}, \mathcal{D}_\p)$
        
        \If{$v = f(\p)$ and $t = 0$}{
            continue
        }
        
        \If{$v = f(\p)$}{
            remove bigrade $(f(\p),0)$ from $\mathsf{B}^+_\p$
        }
        
        \If{$t < 0$}{
            append bigrade $(v, t)$ to $\mathsf{B}^+_\p$
        }
    }        
    append bigrade $(\max(\V), -\infty)$ to $\mathsf{B}^+_\p$
  }
  \caption{Conversion}
  \label{ConversionAlg}
\end{algorithm}

\subsection{Proof of Correctness}\label{4cCorrectness}

First, we prove that our thickening algorithm computes all positive entrance points for all pixels. 
Next, we prove that our thinning algorithm computes all complement entrance points for all pixels.
Together with \cref{thinningAndComplement}, this implies that we correctly compute all negative entrance points for all pixels. 

\begin{theorem} \label{4conThickening}
    \Cref{Thick4C} correctly and exhaustively finds all positive entrance points for all pixels. 
    That is, the algorithm results in $\mathsf{B}_\p^+=\mathcal{B}_\p^+$ for all pixels $\p$.
\end{theorem}

\begin{proof}
    For any $v$, the breadth-first search expands outward from pixels with value $v$: whenever a pixel is popped from the queue, the search moves to its horizontal and vertical neighbors that have not yet been considered at function value $v$.
    Each value $T_{\q,v}$ is the taxicab distance from pixel $\q$ to a closest pixel with value $v$.
    (Pixel $\q$ may have multiple closest pixels with value $v$.)
    Thus, if bigrade $(v, T_{\p,v})$ is added to $\mathsf{B}^+_\p$, then pixel $\p$ is present at that bigrade $(v, T_{\p,v})$ by \Cref{presentPos}.
    
    Furthermore, the bigrades $(v, T_{\p,v})$ added to $\mathsf{B}^+_\p$ are minimal in $\R^+_\p$. 
    If the algorithm adds bigrades $(v, T_{\p,v})$ to $\mathsf{B}^+_\p$, then pixel $\p$ is not present at any distance less than $T_{\p,v}$ for value $v$. 
    Also, pixel $\p$ is not present at distance $T_{\p,v}$ for any $v' < v$ (due to the if statement in the algorithm).
    Thus, $(v, T_{\p,v})$ is a minimal point in $\R^+_\p$. 
    
    Lastly, the algorithm adds \emph{every} minimal point in $\R^+_\p$ to $\mathsf{B}^+_\p$. 
    If $(v,t)$ is a minimal point in $\R^+_\p$, then there exists a pixel $\q$ with value $v$ whose taxicab distance from $\p$ is $t$, and no pixel closer to $\p$ has value less than or equal to $v$.
    When \Cref{Thick4C} runs at value $v$, its breadth-first search encounters pixel $\p$ with $T_{\p,v} = t$, and $(v, T_{\p,v})$ is added to $\mathsf{B}^+_\p$. 
    
    Therefore, \Cref{Thick4C} computes $\mathsf{B}_\p^+$ containing exactly the set of minimal points in $\R^+_\p$, so at the conclusion of the algorithm, $\mathsf{B}_\p^+=\mathcal{B}_\p^+$ for all $\p$.
\end{proof}

We now turn to \cref{Thin4C}, our thinning algorithm. 
For this, recall the duality between thickening and thinning.
\Cref{Thick4C} finds minimal points in $\R^+_\p$ by iterating over $\V$ in increasing order while storing entrance points with positive offset values $t$ that decrease to zero.
Similarly, \Cref{Thin4C} finds maximal points in $\R^c_\p$ by iterating over $\V$ in decreasing order while storing negative offset values $t$ that increase to zero.

\begin{theorem} \label{4conThinning}
    \Cref{Thin4C} correctly and exhaustively finds all complement entrance points for all pixels. 
    That is, the algorithm results in $\mathsf{B}_\p^c=\mathcal{B}_\p^c$ for all $\p$.
\end{theorem}

\begin{proof}
The breadth-first search in \Cref{Thin4C} is similar to that in \Cref{Thick4C}, in that for each pixel $\q$ it finds the nearest pixel with value $v$.
The value $T_{\q,v}$ is the taxicab distance from pixel $\q$ to the closest pixel $\p$ of value $v$, which is equal to the depth of the breadth-first search from $\p$ when $\q$ is encountered.
If $\p$ is closer to $\q$ than any pixel of value greater than $v$, then bigrade $(v-1, -T_{\q,v})$ is maximal in $\R^c_\q$ and thus a complement entrance point for $\q$.
The if statement in \Cref{Thin4C} performs this check; if true, the algorithm appends bigrade $(v-1, -T_{\q,v})$ to $\mathsf{B}^c_\q$ during the breadth-first search for value $v$.
Thus, all points in $\mathsf{B}_\q^c$ are complement entrance points for pixel $\q$.

Furthermore, every complement entrance point $(v-1,-t)$ for pixel $\q$ is exhibited by some pixel $\p$ of value $v$ such that $\p$ is the closest pixel to $\q$ of value $v$ or greater and $d_T(\p,\q)=t$.
The breadth-first search in \cref{Thin4C} finds all such pixels, and thus all complement entrance points for each pixel $\q$.

Thus, \cref{Thin4C} results in $\mathsf{B}_\p^c=\mathcal{B}_\p^c$ for all pixels $\p$.
\end{proof}

Lastly, our conversion of complement entrance points to negative entrance points is correct by \cref{thinningAndComplement}.
We compute the negative entrance point between each pair of consecutive complement entrance points.
However, computing the first negative entrance point (with smallest $v$) requires special care.
In particular, for pixels with value $\min(\V)$, we compute the offset $t$ such that $(\min(\V),t)$ is a negative entrance point, as described previously.
This, together with $(\max(\V), -\infty)$, completes $\mathcal{B}^-_\p$.

If pixel $\p$ has no negative entrance point $(f(\p),t)$ with $t < 0$, then $\mathcal{B}^+_\p \cup \mathcal{B}^-_\p = \mathcal{B}_\p$.
Otherwise, we must remove the positive entrance point $(f(\p),0)$ before taking the union.
Regardless, at the conclusion of \cref{ConversionAlg}, the set $\mathsf{B}^+_\p$ is $\mathcal{B}_\p$, the set of minimal bigrades at which $\p$ is present in the value-offset bifiltration, for each pixel $\p$.

\subsection{Complexity Analysis}

Our thickening algorithm, \Cref{Thick4C}, iterates over $v \in \mathcal{V}$. 
At each iteration, we push pixels into a queue and pop them out until the queue is empty. 
To determine the complexity of this algorithm, we note that a pixel $\p$ is pushed into the queue only after a entrance point $(v,t)$ is found for $\p$.
Since each pixel has at most one entrance point for each value $v$, we see that $\p$ gets added to the queue at most once for each value $v$.
Hence, the while loop runs a maximum of $N$ times.
With each iteration of the while loop, the inner foreach loop runs at most $4$ times.
Thus, the entire thickening process requires $O(N|\mathcal{V}|)$ operations. 

The runtime complexity of the thinning algorithm, \Cref{Thin4C}, is the same: it also requires $O(N|\mathcal{V}|)$ operations.

Our conversion algorithm, \Cref{ConversionAlg}, iterates over all $N$ pixels.
For each pixel, the algorithm must iterate over all bigrades in $\mathsf{B}_\p^c$, of which there at most $|\mathcal{V}|$.
Thus, the conversion algorithm performs $O(N|\mathcal{V}|)$ operations.

Therefore, the runtime complexity for computing the value-offset bifiltration with respect the taxicab distance is $O(N|\mathcal{V}|)$.
This is optimal, since the entrance set may contain $|\mathcal{V}|$ bigrades for each of the $N$ pixels.

The space complexity of our algorithm is dominated by the amount of memory required to store the entrance sets for all pixels. Our algorithm requires $O(N|\mathcal{V}|)$ memory, which is again the size of the bifiltration.

\section{Euclidean Distance}\label{euclideanSection}

Our algorithm from the previous section can be modified to compute a close approximation of the value-offset bifiltration with respect to the Euclidean distance, which we make precise in this section. 

\subsection{Algorithm}

Our algorithm again relies on a breadth-first search for each value $v$.
However, unlike the taxicab distance, the depth of the search process does not correspond to the Euclidean distance between a pixel $\p$ and its nearest pixel with value $v$.
Still, we must guarantee that we remove pixels from the queue in order of increasing distance from pixels with a given value $v$.
We do this by two modifications.
First, we regard pixels as adjacent horizontally, vertically, or \emph{diagonally} (i.e., we regard pixels as being $8$-connected).
Second, we use a priority queue that not only stores pixels $\p$, but also a nearest pixel of value $v$ to each $\p$.
Specifically, elements in our priority queue are triples $(\p, \r, t)$ where $\p$ is a pixel, $\r$ is a closest pixel to $\p$ with value $v$, and $t = d_E(\p,\r)$ is the Euclidean distance between $\p$ and $\r$.
We refer to $\r$ as a ``root'' for $\p$ at value $v$.
The priority queue is kept sorted by the distance values.

Unfortunately, a breadth-first search from root pixels of value $v$ may sometimes fail to associate a pixel $\p$ with the closest root.
To explain this, consider \cref{VoronoiExample}. Suppose the three circled pixels ($\r_1$, $\r_2$, and $\r_3$) are the only root pixels. 
Recall that given a set of points called \emph{sites}, a \emph{Voronoi diagram} partitions the plane into convex regions, each of which contains all points closest to a particular site.
Regarding pixels $\r_1$, $\r_2$, and $\r_3$ as Voronoi sites, the diagonal lines in \cref{VoronoiExample} partition the domain into the three corresponding Voronoi regions.
The shaded pixels---whose centers lie in the middle Voronoi region---comprise a \emph{digital Voronoi region}. Notably, this digital Voronoi region is not $8$-connected, due to the shaded pixel in the upper right.
The possibility of disconnected pixels in digital Voronoi regions is noted elsewhere, such as by Cao et al., who call these disconnected pixels \emph{debris pixels} and observe that they only occur in sharp corners of Voronoi regions \cite{CaoEdelsTan}.

Specifically, debris pixels only occur if a Voronoi region extends between the centers of two adjacent 4-connected pixels to encompass the centers of one or more debris pixels. 
For example, in \cref{VoronoiExample} the column of pixels indicated by the arrow contains no pixel in the shaded digital Voronoi region.
The Voronoi region for site $\r_2$ passes between two adjacent pixels in this column.
In order to pass between the centers of two adjacent pixels, the portion of the Voronoi region containing the debris pixels must have width less than one unit.
Thus, the center of any debris pixel must be within a distance of $\frac{1}{2}$ from a Voronoi edge, a fact that will be important for the proof of \cref{DC_BFS_debris}.

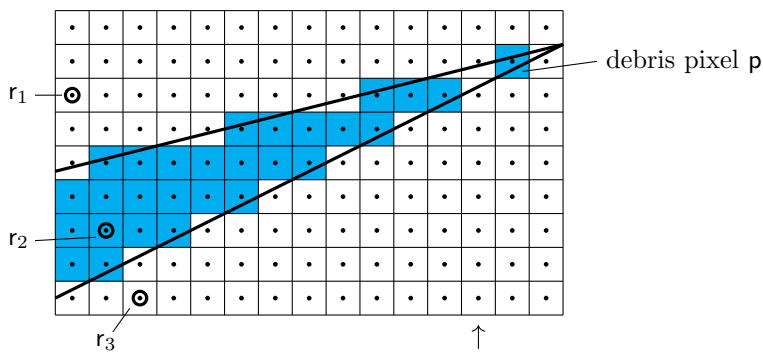
\begin{figure}[hb]
  \centering
  \begin{tikzpicture}[scale=0.45]
    \fill[cyan] (0,1) -- (2,1) -- (2,2) -- (4,2) -- (4,3) -- (6,3) -- (6,4) -- (8,4) -- (8,5) -- (10,5) -- (10,6) -- (12,6) -- (12,7) -- (9,7) -- (9,6) -- (5,6) -- (5,5) -- (1,5) -- (1,4) -- (0,4) -- cycle;
    \fill[cyan] (13,7) rectangle ++(1,1);

    \foreach \x in {0,1,...,15}{
      \draw (\x,0) -- ++(0,9);
    }
    \foreach \y in {0,1,...,9}{
      \draw (0,\y) -- ++(15,0);
    }
    
    \foreach \x in {0,1,...,14}{
      \foreach \y in {0,1,...,8}{
        \fill (\x+0.5,\y+0.5) circle (0.07);
      }
    }
    
    \draw[very thick] (2.5,0.5) circle (0.2);
    \draw[very thick] (1.5,2.5) circle (0.2);
    \draw[very thick] (0.5,6.5) circle (0.2);
    
    
    \draw[very thick] (0,0.5) -- (15,8);
    \draw[very thick] (0,4.25) -- (15,8);
    
    \node[right] at (16,7.5) {debris pixel $\p$};
    \draw (13.8,7.2) -- (16,7.5);
    \node[left] at (-0.5,6.5) {$\r_1$};
    \draw (-0.6,6.5) -- (0.2,6.5);
    \node[left] at (-0.5,2.2) {$\r_2$};
    \draw (-0.6,2.2) -- (1.2,2.3);
    \node at (1.5,-0.8) {$\r_3$};
    \draw (1.7,-0.5) -- (2.2,0.2);
    
    \draw[->] (12.5,-1) -- ++(0,0.7);
  \end{tikzpicture}
  \caption{The set of pixels whose centers lie in a Voronoi region might not be $8$-connected, as the shaded region demonstrates.}
  \label{VoronoiExample}
\end{figure}

A breadth-first search through adjacent pixels, starting with the three root pixels in \cref{VoronoiExample}, will fail to associate the debris pixel $\p$ with the root $\r_2$. Instead, the debris pixel will be encountered from root $\r_1$ or $\r_3$, despite these being at slightly greater distance than $\r_2$.
(In this example, $d_E(\p,\r_1) = d_E(\p,\r_3) = \sqrt{170}$, while $d_E(\p,\r_2) = 13$.)
Worse, a debris pixel may be arbitrarily far from the other pixels in the digital Voronoi region if the Voronoi region is sufficiently narrow.
Thus, it seems quite difficult to correctly compute the distance from a debris pixel to the closest root using a depth-first search. 
Our solution is to instead compute an approximate value-offset bifiltration, allowing small offset errors in entrance points for debris pixels.
This solution is similar to the approach of Cao et al., who produce a modified digital Voronoi diagram to guarantee that each digital Voronoi region is connected \cite{CaoEdelsTan}.

To explain our solution, consider the Voronoi diagram whose sites are all pixels with value $v$.
If there exists a path of $8$-connected pixels from $\p$ to a closest site $\r$ within the digital Voronoi region for $\r$, then we say that $\p$ is a \emph{non-debris pixel} at value $v$. 
In this case, \cref{ThickDC} below finds the correct entrance point for $\p$ at value $v$.
If $\p$ is a non-debris pixel at all $v$, then our algorithm finds the correct entrance set for $\p$.

However, if $\p$ is a debris pixel for some $v$, our breadth-first search computes an approximate entrance point whose offset is incorrect, but the error is less than one unit of distance.
For such a pixel, our algorithm computes entrance points that determine a region $\R'_\p$ which approximates $\R_\p$ in the following sense: for any value $v$, let $t_v = \min\{t \mid (v,t) \in \R_\p\}$ and $t'_v = \min\{t \mid (v,t) \in \R'_\p\}$; then we guarantee that $|t_v - t'_v| < 1$.
If $\mathcal{B}'_\p$ is the set of minimal points in a region $\R'_\p$ such that $|t_v - t'_v| < 1$ for all $v$, we say $\mathcal{B}'_\p$ is an \emph{approximate entrance set} for $\p$. 

We note that the error in our approximate entrance set is both small and rare.
The error in any offset distance is at most one unit, which is the same as the side length of a pixel. Thus, we consider this error to be small.
Furthermore, debris pixels only arise for very particular configurations of root pixels, which do not seem to occur often in digital images, so we consider the error to be rare.

\Cref{ThickDC} gives pseudocode for our thickening algorithm with respect to the Euclidean distance, which approximates the positive entrance sets for each pixel.
For each value $v$, we prepare for the breadth-first search: for each pixel $\r$ of value $v$, we push $(\r, \r, 0)$ into our priority queue, indicating that $\r$ is its own root.
Whenever a triple $(\p, \r, t)$ is removed from the queue, we check whether $(v,t)$ is an entrance point for $\p$.
If so, we append it to the working entrance set $\mathsf{B}^+_\p$ and push into the priority queue all pixels adjacent (horizontally, vertically, or diagonally) to $\p$ along with their distances to $\r$.
Specifically, each pixel $\q$ adjacent to $\p$ is pushed into the queue as a triple $(\q, \r, d_E(\q,\r))$.

\begin{algorithm}[h]
  \ForEach{$v\in\mathcal{V}$ in ascending order}{
    \tcp{prepare breadth-first search at value $v$}
    
    initialize empty queue $Q$
    
    \ForEach{pixel $\r$ such that $f(\r) = v$}{
        $Q$.push($(\r, \r, 0)$) 
    }
    
      \tcp{perform breadth-first search}
      \While{$Q$ is not empty}{
        
        $(\p, \r, t) \gets Q.\mathrm{pop}()$

        \If{$\mathsf{B}^+_\p=\emptyset$ or $t < \min{\{y \mid (x,y)\in \mathsf{B}^+_\p \}}$}{
            
            append $(v, t)$ to $\mathsf{B}^+_\p$ and mark $\p$ as visited
            
            \ForEach{unvisited pixel $\q$ adjacent to $\p$}{
                $Q.\mathrm{push}(\q, \r, d_E(\q, \r))$
            }
        }
      }
   } 
  \caption{Thickening with respect to Euclidean Distance}
  \label{ThickDC}
\end{algorithm}

In contrast to \cref{Thick4C}, \cref{ThickDC} cannot determine whether a pixel $\q$ has an entrance point the first time $\q$ is encountered at value $v$. For this reason, \cref{ThickDC} appends bigrades to $\mathsf{B}^+_\p$ when removing triples from the queue. Furthermore, \cref{ThickDC} remembers which pixels have been ``visited,'' in the sense that an entrance point has been found, to avoid pushing these pixels into the queue again at value $v$.

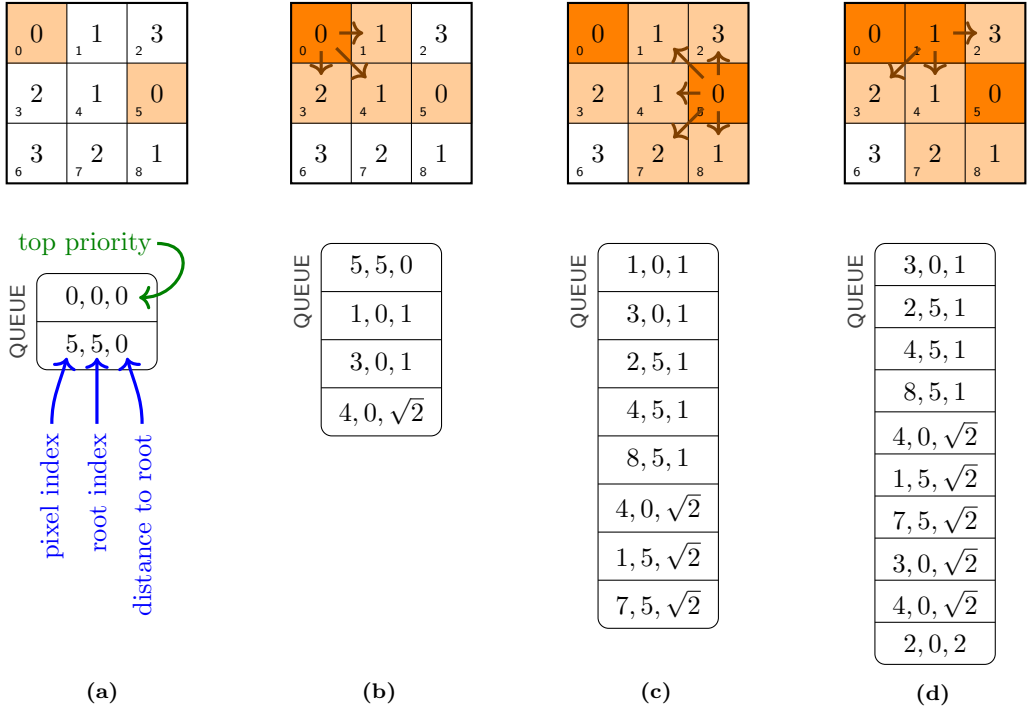
\begin{figure}[h!]
  \input{euclideanDistExample}
  \caption{Illustration of \Cref{ThickDC} for a $3\times3$ image, showing initialization (a) and three iterations of the while loop (b, c, d).}
  \label{EuclideanDistExample}
\end{figure}

\begin{example}
We illustrate \cref{ThickDC} via a simple example displayed in \cref{EuclideanDistExample}. 
The input image consists of a $3\times 3$ image, shown at the top of \cref{EuclideanDistExample}.
Each pixel has a value in $\{0,1,2,3\}$, indicated by the large integers centered in each pixel.
We label the pixels $\p_0, \p_1, \ldots, \p_8$; the index of each pixel is displayed as a small integer in its lower-left corner.

The algorithm prepares for the breadth-first search at value $v=0$.
Each pixel with this value is a root for the breadth-first search and is added to the queue. 
Pixels added to the queue are shaded light orange in \cref{EuclideanDistExample}.
For simplicity, the priority queue illustrated in \cref{EuclideanDistExample} displays pixel indexes and distance value; e.g., $(\p_5, \p_5, 0)$ is shown as $5, 5, 0$.

The algorithm then enters the while loop, where the highest-priority entry is removed from the queue.
The first triple removed is $(\p_0, \p_0, 0)$.
This pixel has entrance point $(v, t) = (0,0)$, which we append to $\mathcal{B}^+_{\p_0}$.
Pixel $\p_0$ is marked as visited, which is indicated in \cref{EuclideanDistExample}(b) by dark orange shading.
We then check all neighbors of pixel $\p_0$; that is, pixels $\p_1$, $\p_2$, and $\p_3$, as illustrated in \cref{EuclideanDistExample}(b).
Since none of them have been visited, we add them to the queue, each with root $\p_0$.

Since the queue is not empty, we remove the top entry from the queue; that is, we remove $(\p_5,\p_5,0)$.
We append entrance point $(0,0)$ to $\mathsf{B}^+_{\p_5}$ and mark $\p_5$ as visited, as indicated in \cref{EuclideanDistExample}(c).
We then check all neighbors of pixel $\p_5$; that is, pixels $\p_1$, $\p_2$, $\p_4$, $\p_7$, and $\p_8$. 
Since none of these neighbors have been visited, they are all added to the queue.

We next remove entry $(\p_1,\p_0,1)$ from the queue. 
This triple tells us that pixel $\p_1$ is distance $1$ from the closest pixel with value $0$ (specifically, $\p_0$), so we append entrance point $(0,1)$ to $\mathsf{B}^+_{\p_1}$ and mark $\p_1$ as visited.
We then check all neighbors of pixel $\p_1$; three of these neighbors are unvisited, so we add them to the queue.
Note that pixel $\p_1$ appears again in the queue, in triple $(\p_1, \p_5, \sqrt{2})$, which tells us that $\p_1$ is distance $\sqrt{2}$ from $\p_5$. When this triple is removed from the queue, no additional bigrade will be added to $\mathsf{B}^+_{\p_1}$.

At this point, we have found entrance points for pixels $\p_0$, $\p_5$ and $\p_1$ at value $v=0$.
The algorithm continues until the queue is empty, and then repeats this process for other values of $v$.
\end{example}

\Cref{ThinDC}, our thinning algorithm with respect to the Euclidean distance, is similar to \cref{ThickDC}. We maintain a priority queue in which each element is a triple of values: a pixel, a root, and the distance between the pixel and root. This priority queue is kept sorted by the distance values.
Distance values are negated in bigrades appended to $\mathsf{B}^c_\p$.
This approximates the complement entrance sets for each pixel.

\begin{algorithm}[h]
  \ForEach{$v \in \V$ such that $v > \min(\V)$ in descending order}{
    \tcp{prepare for breadth-first search at value $v$}
    
    initialize empty queue $Q$
    
    \ForEach{pixel $\p$ such that $f(\p) = v$}{
        append bigrade $(v-1, 0)$ to $\mathsf{B}^c_\p$
        
        $Q.\mathrm{push}((\p, \p, 0))$
    }
    
    \tcp{perform breadth-first search}
    \While{queue $Q$ is not empty}{
        $(\p, \r, t) \gets Q.\mathrm{pop}()$

        \If{$\mathsf{B}^+_\p = \emptyset$ or $t < \min{\{-y \mid (x,y)\in \mathsf{B}^c_\p \}}$}{
            append $(v-1, -t)$ to $\mathsf{B}^c_\p$ and mark $\p$ as visited
            
            \For{unvisited pixel $\q$ adjacent to $\p$}{
                $Q.\mathrm{push}((\q, \r, d_E(\q, \r)))$
            }
        }
      }
   } 

  \caption{Thinning with respect to Euclidean Distance} \label{ThinDC}
\end{algorithm}

Finally, we apply our conversion algorithm, \cref{ConversionAlg}, to produce an approximate entrance set for each pixel, using the output of \cref{ThickDC} and \cref{ThinDC}.
(Recall that \cref{ConversionAlg} does not depend on a particular notion of distance.)

\subsection{Proof of Approximation}

We now establish the correctness of our algorithm for non-debris pixels and establish a bound on the error for debris pixels.

\begin{lemma}\label{DC_BFS_nondebris}
    If $\p$ is a non-debris pixel at value $v$, with entrance point $(v,t)$, then \cref{ThickDC} appends this entrance point to $\mathsf{B}^+_\p$.
\end{lemma}
\begin{proof}
    Let $\r$ be a root pixel for the digital Voronoi region containing $\p$.
    We prove this lemma by induction on the distance from $\r$.
    If $\p$ is adjacent to $\r$, then the lemma is trivially true.
    
    Otherwise, assume the lemma is true for all pixels $\q$ with $d_E(\q,\r) < d_E(\p,\r)$.
    Since Voronoi regions are convex and $\p$ is a non-debris pixel at value $v$, there is a pixel $\q$ adjacent to $\p$ in the same digital Voronoi region for $\r$ such that $d_E(\q,\r) < d_E(\p,\r)$.
    By the inductive hypothesis, the entrance point for $\q$ at value $v$ is found and $(\p, \r, d_E(\p,\r))$ is added to the queue. Since $\p$ is not closer to any other root with value $v$, this ensures that the algorithm discovers the entrance point.
\end{proof}

\begin{lemma}\label{DC_BFS_debris}
    Let $\p$ be a debris pixel at value $v$ that is closest to root $\r_2$. \Cref{ThickDC} discovers $\p$ from a root $r_1$ such that $d_E(\p, \r_1) < d_E(\p,\r_2) + 1$.
\end{lemma}
\begin{proof}
    Since $\p$ is a debris pixel, the center of $\p$ is within $\frac{1}{2}$ unit from a Voronoi edge between $\r_2$ and some other root $\r_1$, as noted earlier.
    That is, there is a point $\nu$ on the Voronoi edge such that $d_E(\p, \nu) < \frac{1}{2}$. (See \cref{equidistant}.)

    Since $\nu$ lies on the Voronoi edge, $\nu$ is equidistant from $\r_1$ and $\r_2$.
    Let $\alpha = d_E(\nu, \r_1) = d_E(\nu, \r_2)$.
    By the triangle inequality,
    \[ d_E(\p, \r_1) \le d_E(\p, \nu) + d_E(\nu, \r_1) < \frac{1}{2} + \alpha, \]
    and also
    \[ d_E(\p, \r_2) \le d_E(\p, \nu) + d_E(\nu, \r_2) < \frac{1}{2} + \alpha. \]
    Applying the triangle inequality again,
    \[ \left| d_E(\p, \r_1) - d_E(\p, \r_2) \right| \le \left| d_E(\p, \r_1) - \alpha \right| + \left| d_E(\p, \r_2) - \alpha \right| < \frac{1}{2} + \frac{1}{2} = 1. \]
    Since $d_E(\p, \r_1) > d_E(\p, \r_2)$, the result follows.
\end{proof}
\begin{figure}[htb]
  \centering
  \begin{tikzpicture}[scale=0.4]
    \foreach \x in {0,1,...,15}{
      \draw[gray] (\x,0) -- ++(0,9);
    }
    \foreach \y in {0,1,...,9}{
      \draw[gray] (0,\y) -- ++(15,0);
    }

    \foreach \x in {0,1,...,14}{
      \foreach \y in {0,1,...,8}{
        \fill (\x+0.5,\y+0.5) circle (0.07);
      }
    }
    
    \fill (13.4706, 7.6176) circle (0.08);
    
    \draw[very thick] (2.5,0.5) circle (0.2);
    \draw[very thick] (1.5,2.5) circle (0.2);
    \draw[very thick] (0.5,6.5) circle (0.2);
    
    
    \draw[very thick] (0,0.5) -- (15,8);
    \draw[very thick] (0,4.25) -- (15,8);
    
    \draw[orange,very thick] (13,7) circle (1.3);
    
    \node[left] at (-0.5,6.5) {$\r_1$};
    \draw (-0.6,6.5) -- (0.2,6.5);
    \node[left] at (-0.5,2.2) {$\r_2$};
    \draw (-0.6,2.2) -- (1.2,2.3);
    \node at (1.5,-0.8) {$\r_3$};
    \draw (1.7,-0.5) -- (2.2,0.2);
    
    \draw[gray] (22,1) -- (22,9);
    \draw[gray] (25.077,2.444) -- (25.077,7.556);
    \draw[gray] (18.923,2.444) -- (18.923,7.556);
    \draw[gray] (18,5) -- (26,5);
    \draw[gray] (19.444,1.923) -- (24.556,1.923);
    \draw[gray] (19.444,8.077) -- (24.556,8.077);
    
    \fill (23.538,6.538) circle (0.15);
    \fill (20.462,6.538) circle (0.15);
    \fill (20.462,3.462) circle (0.15);
    \fill (23.538,3.462) circle (0.15);
    
    \draw[very thick] (18.422,3.211) -- (25.578,6.789);
    \draw[very thick] (18.038,5.548) -- (25.239,7.348);
    
    \fill (23.448, 6.9) circle (0.15);
    
    \node at (23.1,6.1) {$\p$};
    \node at (23.2,7.3) {$\nu$};
    
    \draw[orange,very thick] (22,5) circle (4);
    \draw[orange,very thick] (12.9,8.29) -- (21.7,8.98);
    \draw[orange,very thick] (12.3,5.9) -- (20.2,1.42);
  \end{tikzpicture}
  \caption{Voronoi edges between sites $\r_1$, $\r_2$, and $\r_3$ are shown. Pixel $\p$ is a debris pixel. Point $\nu$ is equidistant from $\r_1$ and $\r_2$.}
  \label{equidistant}
\end{figure}
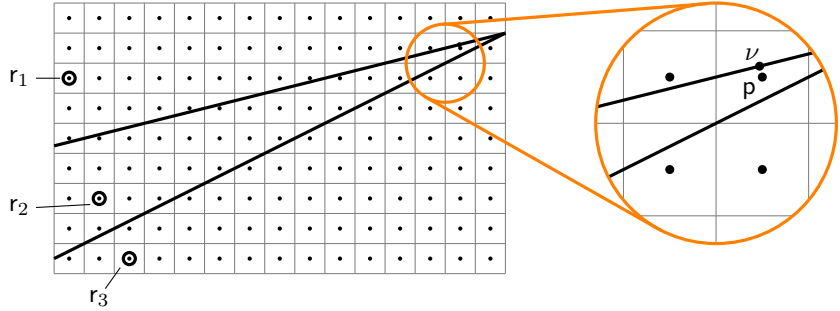

We now show that applying \cref{ThickDC} and \cref{ThinDC}, followed by \cref{ConversionAlg}, produces an approximate entrance set for each pixel. This gives an approximation of the value-offset bifiltration with respect to the Euclidean distance such that, for any value $v$, the offset $t$ at which a pixel appears in the bifiltration is within $1$ unit of the correct offset.

\begin{theorem}
    The sequence of algorithms \Cref{ThickDC}, \Cref{ThinDC}, and \Cref{ConversionAlg} compute an approximate entrance set for each pixel.
\end{theorem}
\begin{proof}
    Fix a value $v$ and a pixel $\p$ that has an entrance point $(v,t)$.
    
    If $\p$ is a non-debris pixel at value $v$, then by \cref{DC_BFS_nondebris}, \cref{ThickDC} computes the correct entrance point $(v,t)$.
    If $\p$ is a debris pixel at value $v$, then by \cref{DC_BFS_debris}, \cref{ThickDC} computes an approximate entrance point $(v,t')$ such that $t < t' < t+1$, which is appended to $\mathsf{B}^+_\p$ if there does not exist $\mathbf{b} \in \mathsf{B}^+_\p$ with $\mathbf{b} \prec (v,t')$.
    In either case, \cref{ThickDC} returns an approximate positive entrance set.
    
    Similarly, \cref{ThinDC} returns an approximate complement entrance set. 
    Specifically, if $\p$ has a complement entrance point $(v,t)$, then the approximate complement entrance point $(v,t')$ computed by \cref{ThinDC} satisfies $(t-1 < t' \le t)$.
    Applied to this approximate complement entrance set, \cref{ConversionAlg} computes an approximate negative entrance set satisfies the same property: if $\p$ has a negative entrance point $(v,t)$, then the computed entrance point $(v,t')$ satisfies $(t-1 < t' \le t)$.
    
    Combining the approximate positive and negative entrance sets, as in \cref{ConversionAlg}, produces an approximate entrance set $\mathsf{B}^+_\p$.
\end{proof}
 
\subsection{Complexity Analysis}

Our thickening algorithm performs a breadth-first search for each $v \in \mathcal{V}$.
Since each pixel has at most $8$ neighbors (horizontally, vertically, and diagonally), \Cref{ThickDC} adds each pixel to the queue at most $8$ times for each value $v$.
This implies that the size of the queue is always $O(N)$; consequently each addition and removal from the priority queue is performed in $O(\log N)$ time.
There are $O(N)$ addition and removal operations, so the process of finding all positive entrance points for any particular value $v$ requires $O(N \log N)$ time.
The runtime complexity of \Cref{ThickDC} is therefore $O(N|\mathcal{V}|\log N)$.

Similarly, the runtime complexity of \Cref{ThinDC} is $O(N|\mathcal{V}|\log N)$.
The conversion algorithm requires $O(N|\mathcal{V}|)$ operations as before.
Therefore, the runtime complexity for computing the value-offset bifiltration with respect to Euclidean distance is $O(N|\mathcal{V}|\log N)$.

The space complexity is again dominated by the memory required to store the bifiltration, which is $O(N|\mathcal{V}|)$.

\section{Experimental Results}\label{resultsSection}

\subsection{Real Images vs.\ Generated Images}

We compared our algorithms on both real images and randomly generated images, computing bifiltrations with respect to the taxicab and Euclidean distances.
Our real images are a collection of eleven artistic images, converted to grayscale.
We created a randomly generated image of the same pixel dimensions as each real image, with each pixel color value a randomly selected integer from $0$ to $255$. 
Our collection of images is available in our code repository.\footnote{\url{ https://github.com/ThongVoHien/TopologyAnalysis}; see especially the \textsf{data files} directory.}
Experiments were performed on a Intel Core i5-7200U 2.5GHz processor with 16GB memory. 

\Cref{runtimeComparison} compares the runtimes of our algorithms to the total number of pixels in each image.
While the runtime generally increases with the number of pixels, we observe that runtimes for the algorithm with respect to the Euclidean distance algorithm are roughly ten times as long as for the taxicab distance. 
This is due to the additional complexity of maintaining the priority queue for the Euclidean distance, as well as the fact that the entrance sets are larger with respect to the Euclidean distance. 

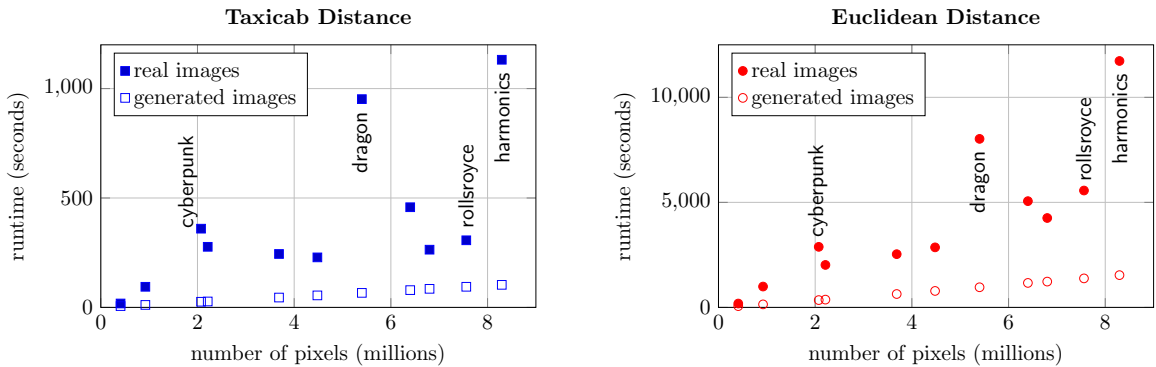
\begin{figure}[h] 
  \centering
    \pgfplotsset{yticklabel style={/pgf/number format/fixed, /pgf/number format/precision=5},scaled y ticks=false}
    \begin{tikzpicture}[scale=0.82]
      \begin{axis}[width=3.4in, height=2.3in, title={\textbf{Taxicab Distance}}, xlabel={number of pixels (millions)}, ylabel={runtime (seconds)}, xmin=0, xmax=9, ymin=0, ymax=1200, grid, legend pos=north west, legend cell align=left]
        \addplot+[only marks,mark=square*,blue] table[x=pixels,y=real4C]{runtimes.data};
        \addlegendentry{real images}
        \addplot+[only marks,mark=square,blue] table[x=pixels,y=gen4C]{runtimes.data};
        \addlegendentry{generated images}
        \node[left,rotate=90] at (8.3,1110) {\textsf{harmonics}};
        \node[right,rotate=90] at (7.6,330) {\textsf{rollsroyce}};
        \node[right,rotate=90] at (1.8,330) {\textsf{cyberpunk}};
        \node[left,rotate=90] at (5.4,930) {\textsf{dragon}};
      \end{axis}
    \end{tikzpicture} \hspace{20pt}
    \begin{tikzpicture}[scale=0.82]
      \begin{axis}[width=3.4in, height=2.3in, title={\textbf{Euclidean Distance}}, xlabel={number of pixels (millions)}, ylabel={runtime (seconds)}, xmin=0, xmax=9, ymin=0, ymax=12500, grid, legend pos=north west, legend cell align=left]
        \addplot+[only marks,mark=*,red,mark options={fill=red}] table[x=pixels,y=realDC]{runtimes.data};
        \addlegendentry{real images}
        \addplot+[only marks,mark=o,red] table[x=pixels,y=genDC]{runtimes.data};
        \addlegendentry{generated images}
        \node[left,rotate=90] at (8.3,11500) {\textsf{harmonics}};
        \node[right,rotate=90] at (7.6,5650) {\textsf{rollsroyce}};
        \node[right,rotate=90] at (2.05,3000) {\textsf{cyberpunk}};
        \node[left,rotate=90] at (5.4,7800) {\textsf{dragon}};
      \end{axis}
    \end{tikzpicture}
  \caption{Runtime comparison of our algorithms on real and generated images. For generated images, runtimes are strongly linear in the number of pixels: for the taxicab distance the slope is $12.4$ seconds per million pixels and the correlation is $0.9999$, while the Euclidean distance yields a slope of $187$ seconds per million pixels and correlation $0.9996$. }
  \label{runtimeComparison}
\end{figure}

We also observe that the runtimes for generated images show a strong linear relationship, while the runtimes for real images are larger and more scattered.
The runtime has to do more with the \emph{structure} of the image than the number of pixels.
In general, the real images exhibit patterns of light and dark pixels that produce large entrance sets (i.e., containing many entrance bigrades) for many pixels.
However, images with color values more uniformly scattered throughout, such as our generated images, produce relatively small entrance sets per pixel.

We offer some specific examples.
The datapoints labeled in \cref{runtimeComparison} correspond with the images displayed in \cref{imageSamples}.
We observe that the images \textsf{cyberpunk}, \textsf{dragon}, and \textsf{harmonics} have higher runtimes than other real images with similar numbers of pixels. 
The images \textsf{cyberpunk} and especially \textsf{dragon} contain a dominant light source that contrasts with darker regions elsewhere in the image. 
This results in large entrance sets, on average, for pixels in these images, as recorded in \cref{avgBigradesReal}.
Images such as \textsf{harmonics} display regular patterns of dark and light pixels, which also results in large entrance sets.
On the other hand, the image \textsf{rollsroyce} has a relatively small runtime despite a large number of pixels. This is because its light and dark pixels produce relatively small entrance sets, as shown in \cref{avgBigradesReal}.

\begin{figure}[h]
    \centering
    \begin{tabular}{cc}
        \includegraphics[height=1in]{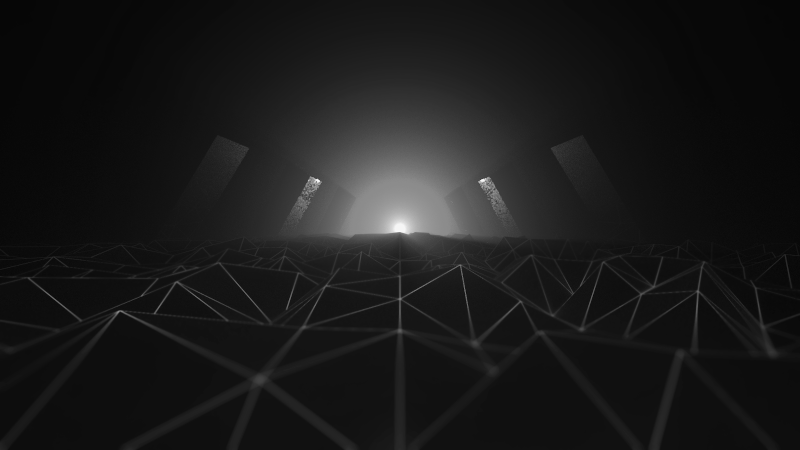} & \includegraphics[height=1in]{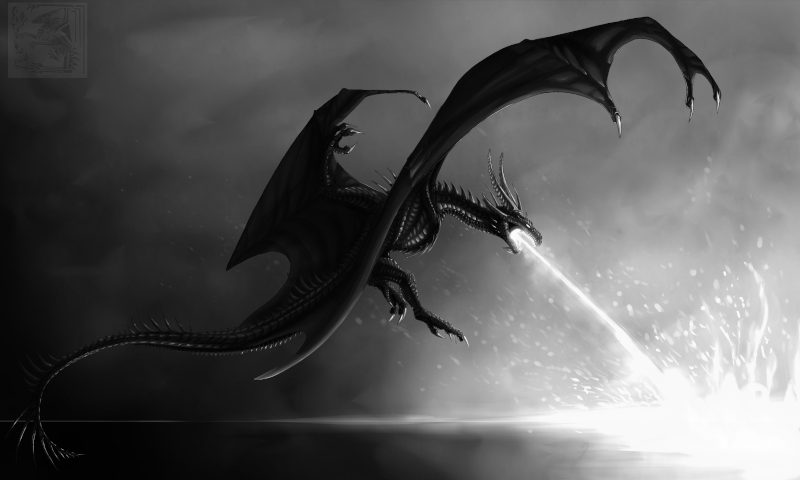}  
        \\
        \textsf{cyberpunk} & \textsf{dragon} \\
        
        \ & \ \\
        
        \includegraphics[height=1in]{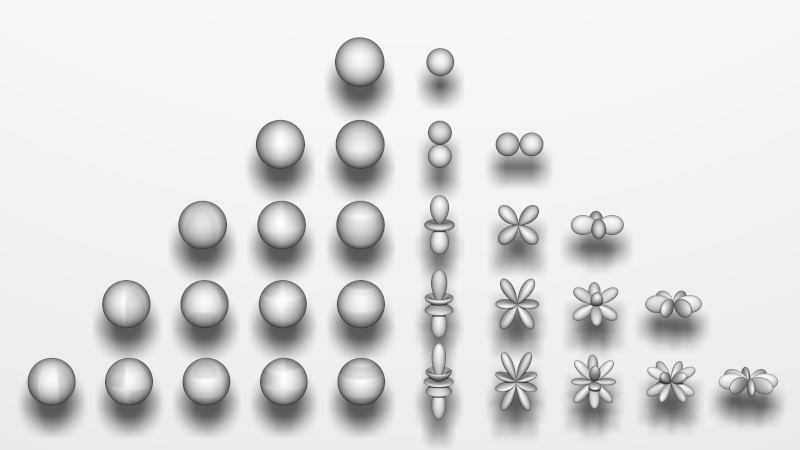} & \includegraphics[height=1in]{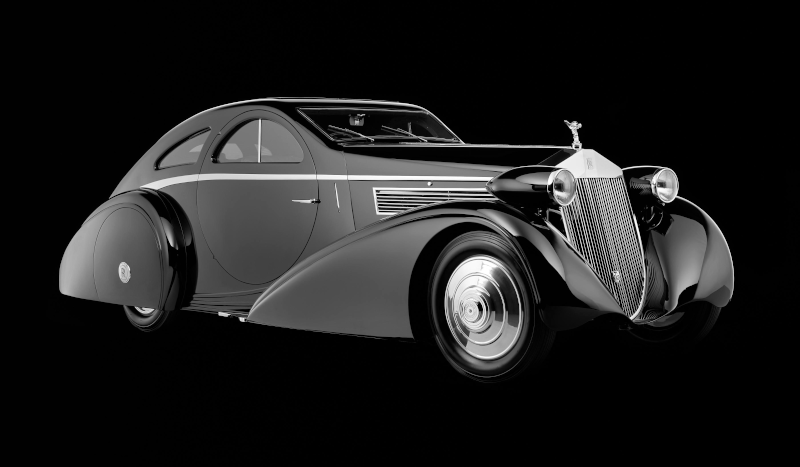} \\
        \textsf{harmonics} & \textsf{rollsroyce}
    \end{tabular}
    \caption[images]{Four of the ``real'' images used in our experiments.\footnotemark }
    \label{imageSamples}
\end{figure}

\begin{table}[h]
    \centering
    \begin{tabular}{lccc}\toprule
        & & \multicolumn{2}{c}{avg.\ num.\ entrance points per pixel} \\ \cmidrule(lr){3-4} 
        image & num.\ pixels & taxicab distance & Euclidean distance \\ \midrule
        generated & --- & 8.95 & 9.67 \\
        \textsf{cyberpunk} & 2,073,600 & 73.3 & 88.4  \\
        \textsf{dragon} & 5,400,000 & 89.0 & 97.6 \\
        \textsf{rollsroyce} & 7,560,000 & 28.0 & 43.9 \\
        \textsf{harmonics} & 8,294,400 & 70.2 & 83.8 \\
        \bottomrule
    \end{tabular}
    \caption{Average number of entrance points per pixel for selected images. The first line displays the averages over eleven generated images of various pixel sizes.}
    \label{avgBigradesReal}
\end{table}

Next we compare the average entrance set sizes with respect to the taxicab and Euclidean distances, as shown in \cref{avgBigradesReal} and \cref{avgNumBigrades}.
We see that the average number of entrance points per pixel is always smaller for the taxicab distance than for Euclidean distance.
The image \textsf{dragon} has the most entrance points per pixel, while \textsf{cyberpunk} and \textsf{harmonics} are not far behind.
However, \textsf{rollsroyce} has the fewest entrance points per pixel with respect to the taxicab distance and second-fewest with respect to the Euclidean distance.
The average number of entrance points per pixel is nearly constant across our set of eleven generated images, with a standard deviation of $0.013$ with respect to the taxicab distance and $0.0098$ with respect to the Euclidean distance.

\footnotetext{ Image sources: \textsf{cyberpunk} \url{https://dlpng.com/png/6389782}, \textsf{dragon} \url{https://wallpapersden.com/dragon-burning-flames-wallpaper}, \textsf{harmonics} \url{https://commons.wikimedia.org/wiki/File:Cubicharmonics_3840x2160.png}, \textsf{rollsroyce} \url{https://www.carthrottle.com/post/wql2dlq/} }

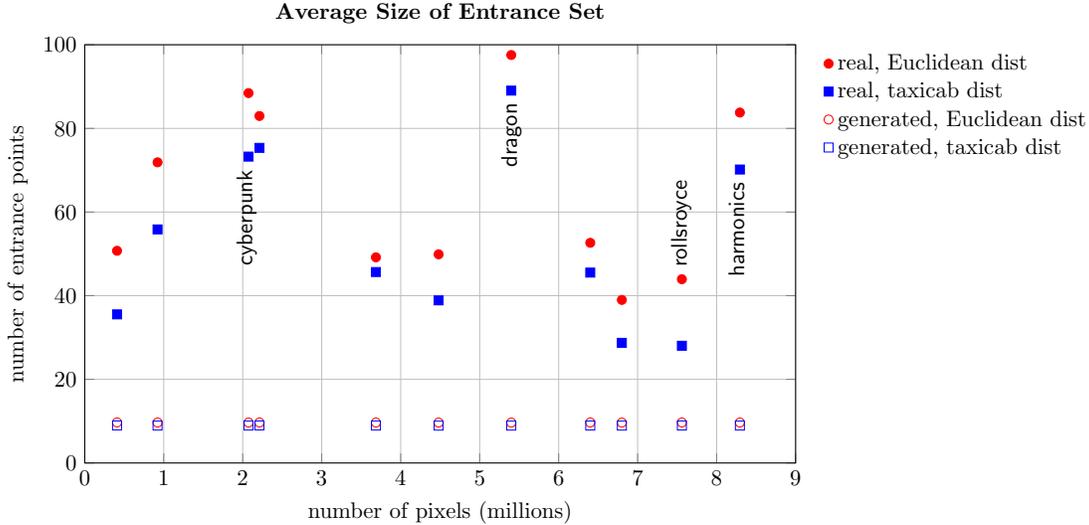
\begin{figure}[h!] 
  \centering
    \begin{tikzpicture}[scale=0.85]
      \begin{axis}[width=5in, height=3.2in, title={\textbf{Average Size of Entrance Set}}, xlabel={number of pixels (millions)}, ylabel={number of entrance points}, xmin=0, xmax=9, ymin=0, ymax=100, grid, legend pos=outer north east, legend cell align=left, legend style={draw=none}]
        \addplot+[only marks,mark=*,red,mark options={fill=red}] table[x=pixels,y=realDC]{entranceSetSizes.data};
        \addlegendentry{real, Euclidean dist}
      
        \addplot+[only marks,mark=square*,blue, mark options={fill=blue}] table[x=pixels,y=real4C]{entranceSetSizes.data};
        \addlegendentry{real, taxicab dist}
        
        \addplot+[only marks,mark=o,red] table[x=pixels,y=genDC]{entranceSetSizes.data};
        \addlegendentry{generated, Euclidean dist}
        
        \addplot+[only marks,mark=square,blue] table[x=pixels,y=gen4C]{entranceSetSizes.data};
        \addlegendentry{generated, taxicab dist}
        
        \node[left,rotate=90] at (8.25,69) {\textsf{harmonics}};
        \node[right,rotate=90] at (7.55,45) {\textsf{rollsroyce}};
        \node[left,rotate=90] at (2.05,72) {\textsf{cyberpunk}};
        \node[left,rotate=90] at (5.4,88) {\textsf{dragon}};
        
      \end{axis}
    \end{tikzpicture}

  \caption{Average number of entrance points per pixel, for real and generated images with respect to taxicab and Euclidean distances.}
  \label{avgNumBigrades}
\end{figure}

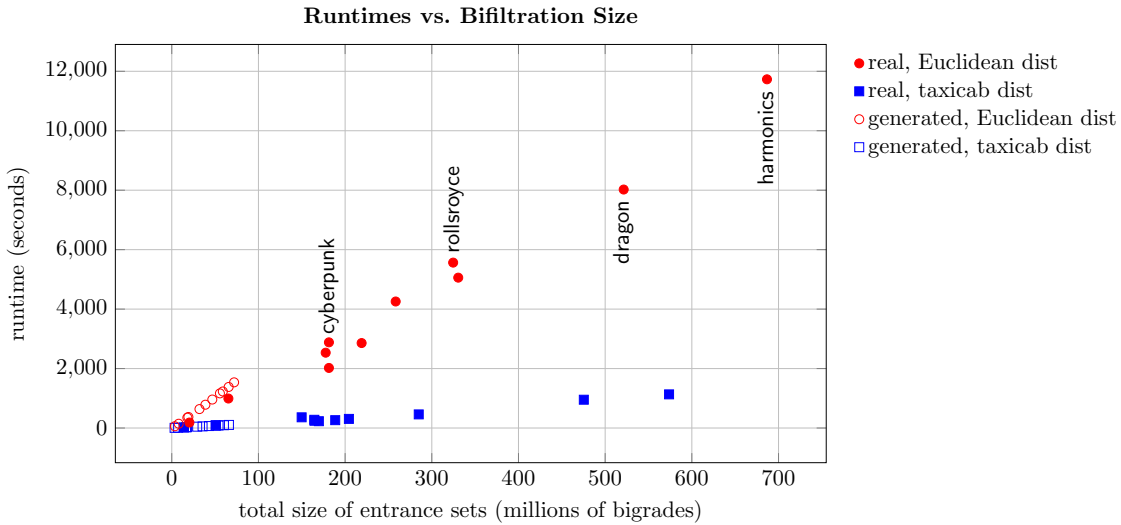
\begin{figure}[h!] 
  \centering
    \begin{tikzpicture}[scale=0.85]
      \begin{axis}[width=5in, height=3.2in, title={\textbf{Runtimes vs.\ Bifiltration Size}}, xlabel={total size of entrance sets (millions of bigrades)}, ylabel={runtime (seconds)}, scaled y ticks=false, grid, legend pos=outer north east, legend cell align=left, legend style={draw=none}, scatter/classes={realDC={mark=*,red},real4C={mark=square*,blue},genDC={mark=o,draw=red},gen4C={mark=square,draw=blue}}]
        \addplot[scatter, only marks, scatter src=explicit symbolic] table[x=bigrades, y=time, meta=type]{runtimesVsBigrades.data};
        \addlegendentry{real, Euclidean dist}
        \addlegendentry{real, taxicab dist}
        \addlegendentry{generated, Euclidean dist}
        \addlegendentry{generated, taxicab dist}
        
        \node[left,rotate=90] at (685, 11600) {\textsf{harmonics}};
        \node[left,rotate=90] at (521, 7940) {\textsf{dragon}};
        \node[right,rotate=90] at (325, 5620) {\textsf{rollsroyce}};
        \node[right,rotate=90] at (181, 2900) {\textsf{cyberpunk}};
      \end{axis}
    \end{tikzpicture}
  \caption{The runtimes of our algorithms are approximately linear in the total number of bigrades in the entrance sets for all pixels in the value-offset bifiltration.}
  \label{runtimeVsBigrades}
\end{figure}

The entrance set sizes explain the differences in runtime.
\Cref{runtimeVsBigrades} shows how our runtimes depend on the total size of the bifiltration (that is, the total number of bigrades for all pixels).
For real and generated images, with each notion of distance, we find that the runtimes are nearly linear in the total number of bigrades.
As before, runtimes of the generated images show the strongest correlation: with respect to the taxicab distance the slope is $1.56$ seconds/million bigrades with correlation $1.0000$; with respect to the Euclidean distance the slope is $21.56$ seconds/million bigrades with correlation $0.9997$.
For real images with respect to the taxicab distance, the slope is $2.02$ seconds/million bigrades with correlation $0.985$; for real images with respect to the Euclidean distance the slope is $17.12$ seconds/million bigrades with correlation $0.993$.

\subsection{Worst Case Examples}\label{worstCase}

Based on our observations that images such as \textsf{dragon} and \textsf{harmonics} have rather large entrance sets per pixel, we found two types of images whose bifiltrations have the largest possible entrance sets for all pixels.

Specifically, we constructed images that have the highest value pixels concentrated in their centers, decreasing to low-value pixels in their edges and corners. \Cref{worstCaseExamples} (a) shows an example of such an image, which we call a \emph{centralized} image.
We also constructed images with constant values along diagonals, such that adjacent diagonals have values that differ by 1, as shown in \Cref{worstCaseExamples} (b). 
We call these \emph{diagonal} images.

\begin{figure}[h]
  \centering
  \begin{subfigure}[b]{0.4\textwidth}
    \centering
    \begin{tikzpicture}[scale=0.5] 
        \fill[black!90] (0,0) rectangle ++(6,6);  
        \fill[black!60] (1,0) rectangle ++(4,6);  
        \fill[black!60] (0,1) rectangle ++(6,4);  
        \fill[black!30] (1,1) rectangle ++(4,4);  
        \fill[white] (2,2) rectangle ++(2,2);     
        
        \draw[thick] (0,0) rectangle (6,6);
        \foreach \i in {1,...,5} {
          \draw (\i,0) -- (\i,6);
          \draw (0,\i) -- (6,\i);
        }
        
        \foreach \p in {(0.5,0.5),(5.5,0.5),(0.5,5.5),(5.5,5.5)} {
          \node[white] at \p {0};
        }
        \foreach \p in {(1.5,0.5),(2.5,0.5),(3.5,0.5),(4.5,0.5),(0.5,1.5),(0.5,2.5),(0.5,3.5),(0.5,4.5),(1.5,5.5),(2.5,5.5),(3.5,5.5),(4.5,5.5),(5.5,1.5),(5.5,2.5),(5.5,3.5),(5.5,4.5)} {
          \node[white] at \p {1};
        }
        \foreach \p in {(1.5,1.5),(2.5,1.5),(3.5,1.5),(4.5,1.5)  ,(1.5,2.5),(1.5,3.5),(4.5,2.5),(4.5,3.5),(1.5,4.5),(2.5,4.5),(3.5,4.5),(4.5,4.5)} {
          \node at \p {2};
        }
        \foreach \p in {(2.5,2.5),(2.5,3.5),(3.5,2.5),(3.5,3.5)} {
          \node at \p {3};
        }
        
    \end{tikzpicture}
    \caption{}
  \end{subfigure}
  \begin{subfigure}[b]{0.4\textwidth}
    \centering
    \begin{tikzpicture}[scale=0.5] 
        \fill[black!30] (0,3) -- (1,3) -- (1,4) -- (2,4) -- (2,5) -- (3,5) -- (3,6) -- (0,6) -- cycle; 
        \fill[black!60] (0,4) -- (1,4) -- (1,5) -- (2,5) -- (2,6) -- (0,6) -- cycle; 
        \fill[black!90] (0,5) rectangle ++(1,1);  
        
        \fill[black!30] (0,0) -- (4,0) -- (4,1) -- (5,1) -- (5,2) -- (6,2) -- (6,3) -- (6,6) -- (4,6) -- (4,5) -- (3,5) -- (3,4) -- (2,4) -- (2,3) -- (1,3) -- (1,2) -- (0,2) -- cycle; 
        
        \fill[black!60] (0,0) -- (3,0) -- (3,1) -- (4,1) -- (4,2) -- (5,2) -- (5,3) -- (6,3) -- (6,6) -- (5,6) -- (5,5) -- (4,5) -- (4,4) -- (3,4) -- (3,3) -- (2,3) -- (2,2) -- (1,2) -- (1,1) -- (0,1) -- cycle;  
        
        \fill[black!90] (1,0) rectangle ++(1,1);  
        \fill[black!90] (2,1) rectangle ++(1,1);  
        \fill[black!90] (3,2) rectangle ++(1,1);  
        \fill[black!90] (4,3) rectangle ++(1,1);  
        \fill[black!90] (5,4) rectangle ++(1,1);  
        
        \fill[black!30] (5,0) rectangle ++(1,1);  
        
        \draw[thick] (0,0) rectangle (6,6);
        \foreach \i in {1,...,5} {
          \draw (\i,0) -- (\i,6);
          \draw (0,\i) -- (6,\i);
        }
        
        \foreach \p in {(0.5,5.5),(1.5,0.5),(2.5,1.5),(3.5,2.5),(4.5,3.5),(5.5,4.5)} {
          \node[white] at \p {0};
        }
        \foreach \p in {(0.5,4.5),(1.5,5.5),(0.5,0.5),(1.5,1.5),(2.5,2.5),(3.5,3.5),(4.5,4.5),(5.5,5.5),(2.5,0.5),(3.5,1.5),(4.5,2.5),(5.5,3.5)} {
          \node[white] at \p {1};
        }
        \foreach \p in {(0.5,3.5),(1.5,4.5),(2.5,5.5),(0.5,1.5),(1.5,2.5),(2.5,3.5),(3.5,4.5),(4.5,5.5),(3.5,0.5),(4.5,1.5),(5.5,2.5),(5.5,0.5)} {
          \node at \p {2};
        }
        \foreach \p in {(0.5,2.5),(1.5,3.5),(2.5,4.5),(3.5,5.5),(4.5,0.5),(5.5,1.5)} {
          \node at \p {3};
        }
        
    \end{tikzpicture}
        \caption{}
  \end{subfigure}
    
  \caption{Examples of \emph{centralized} (a) and \emph{diagonal} (b) images with grayscale values $0, 1, 2, 3$.}
   \label{worstCaseExamples}
\end{figure}
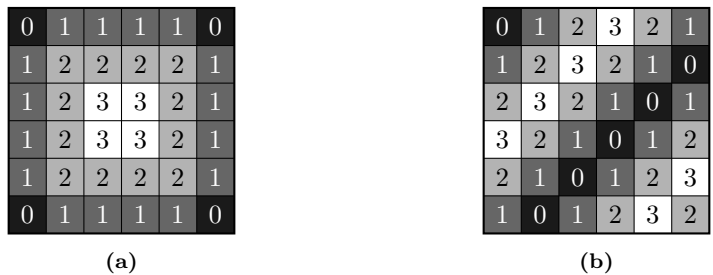

The key property is that the values along any pixel path from a pixel of lowest value to highest value are monotonically increasing. 
This property, exhibited by centralized and diagonal images (as well as many other images), implies that each pixel has an entrance point for each function value.
Since no entrance set can have more bigrades than the number of distinct function values, these images provide worst-case examples of the size of the value-offset bifiltration.

\section{Future work}\label{futureWork}

The goal of this project has been to develop the value-offset bifiltration as a foundation to support two-parameter persistent homology analysis of digital image data. 
As a next step, we plan to compute two-parameter persistent homology of value-offset bifiltrations.
For this, we must incorporate our algorithms from this paper into a two-parameter persistence computational pipeline, such as the RIVET software \cite{rivet}.
The remaining work to this end consists primarily of implementation details rather than the development of theory or algorithms.
The entrance sets computed in this paper are the input required for constructing bifiltered complexes, of which existing algorithms can compute persistent homology.

Further questions remain about the value-offset bifiltration, especially with respect to distances other than the taxicab distance.
For example, how could one compute the exact entrance sets with respect to the Euclidean distance? We do not know an efficient way to do this. 
Alternatively, it would be interesting to consider other metrics on sets of pixels.
Lastly, it seems natural to generalize our work to higher-dimensional cubical data.
For example, we could compute a value-offset bifiltration from three-dimensional voxels with function values, perhaps obtained via 3-D scanning technology. 
This would seem to open new avenues for topological analysis of high-dimensional digital data.

\section*{Code}

The code for our algorithms and experiments described in this paper is available at: \\ \url{https://github.com/ThongVoHien/TopologyAnalysis}

\section*{Acknowledgements}

This research was carried out while the first two authors were undergraduate students at St.\ Olaf College. We thank St.\ Olaf College, and especially the MSCS Department, for their support of undergraduate research.
We also thank the anonymous reviewers, whose comments and suggestions improved this paper.

\end{document}

%% file: bifiltration_example.tex
\begin{tikzpicture}[scale=0.6]
    \fill[black!30] (0,1) rectangle ++(3,3); 
    \fill[black!60] (0,2) rectangle ++(2,2); 
    \fill[black!60] (0,0) rectangle ++(1,1); 
    \fill[black!60] (2,0) rectangle ++(2,1); 
    \fill[black!90] (0,3) rectangle ++(1,1); 
    \fill[black!90] (3,1) rectangle ++(1,1); 
    
    \node[white] at (0.5,3.5) {0};
    \node[white] at (3.5,1.5) {0};
    \foreach \p in {(0.5,0.5),(2.5,0.5),(3.5,0.5),(0.5,2.5),(1.5,2.5),(1.5,3.5)} {
      \node[white] at \p {1};
    }
    \foreach \p in {(0.5,1.5),(1.5,1.5),(2.5,1.5),(2.5,2.5),(2.5,3.5)} {
      \node at \p {2};
    }
    \foreach \p in {(1.5,0.5),(3.5,2.5),(3.5,3.5)} {
      \node at \p {3};
    }
    
    \draw (0,0) rectangle ++ (4,4);
    \foreach \i in {1, 2, 3}{
      \draw (\i,0) -- ++(0,4);
      \draw (0,\i) -- ++(4,0);
    }
\end{tikzpicture}
  
\vspace{4pt}
  
\noindent\rule{0.9\textwidth}{0.8pt}
  
\vspace{16pt}
  
\begin{tikzpicture}[scale=0.44]
    \foreach \x\y in {0/18, 6/12, 6/18, 12/6, 12/12, 12/18, 18/0, 18/6, 18/12, 18/18, 18/-6, 18/-12, 18/-18}{
      \fill[darkgray] (\x,\y) rectangle ++(4,4);
    }
    
    \fill[darkgray] (0,3) rectangle ++(1,1);
    \fill[darkgray] (3,1) rectangle ++(1,1);
    
    \fill[darkgray] (6,2) rectangle ++(2,2);
    \fill[darkgray] (6,0) rectangle ++(1,1);
    \fill[darkgray] (8,0) -- (10,0) -- (10,2) -- (9,2) -- (9,1) -- (8,1) -- cycle;
    
    \fill[darkgray] (12,0) rectangle ++(1,1);
    \fill[darkgray] (12,1) rectangle ++(3,3);
    \fill[darkgray] (14,0) rectangle ++(2,2);
    
    \fill[darkgray] (0,8) -- (1,8) -- (1,9) -- (2,9) -- (2,10) -- (0,10) -- cycle;
    \fill[darkgray] (3,6) -- (4,6) -- (4,9) -- (3,9) -- (3,8) -- (2,8) -- (2,7) -- (3,7) -- cycle;
    
    \fill[darkgray] (6,6) -- (10,6) -- (10,9) -- (9,9) -- (9,10) -- (6,10) -- cycle;
    
    \fill[darkgray] (2,12) -- (4,12) -- (4,16) -- (0,16) -- (0,13) -- (2,13) -- cycle;
    
    \fill[darkgray] (6,-3) rectangle ++(1,1);
    \fill[darkgray] (9,-6) rectangle ++(1,1);
    
    \fill[darkgray] (12,-4) rectangle ++(2,2);
    \fill[darkgray] (12,-5) rectangle ++(1,1);
    \fill[darkgray] (14,-5) rectangle ++(1,1);
    \fill[darkgray] (15,-6) rectangle ++(1,1);
    
    \fill[darkgray] (12,-10) rectangle ++(1,2);
    
    \foreach \x in {0, 6, 12, 18}{
      \foreach \y in {-18, -12, -6, 0, 6, 12, 18}{
        \draw (\x,\y) rectangle ++(4,4);
        \foreach \i in {1, 2, 3}{
          \draw (\x+\i,\y) -- ++(0,4);
          \draw (\x,\y+\i) -- ++(4,0);
        }
      }
    }
    
    \node at (-2,20) {$t=3$};
    \node at (-2,14) {$t=2$};
    \node at (-2,8) {$t=1$};
    \node at (-2,2) {$t=0$};
    \node at (-2,-4) {$t=-1$};
    \node at (-2,-10) {$t=-2$};
    \node at (-2,-16) {$t=-3$};
    
    \node at (2,-.7) {$v=0$};
    \node at (8,-.7) {$v=1$};
    \node at (14,-.7) {$v=2$};
    \node at (20,-.7) {$v=3$};
    
    \draw[green, rounded corners,thick] (-3.5,-1.5) rectangle (23,5);
    
\end{tikzpicture}
  

%% file: imageExample.tex
    \centering
    \begin{subfigure}[b]{0.3\textwidth}
      \centering
      \begin{tikzpicture}[scale=0.5] 
        \fill[black!90] (1,2) rectangle (5,6); 
        \fill[black!75] (0,0) rectangle (6,2); 
        \fill[black!75] (1,3) rectangle (3,4); 
        \fill[black!75] (5,4) rectangle (6,5); 
        \fill[black!60] (0,0) rectangle (2,1); 
        \fill[black!60] (3,1) rectangle ++(1,1); 
        \fill[black!60] (3,3) rectangle ++(1,1); 
        \fill[black!60] (1,4) rectangle ++(1,1); 
        \fill[black!60] (2,5) rectangle ++(1,1); 
        \fill[black!45] (0,2) rectangle ++(1,1); 
        \fill[black!45] (1,1) rectangle ++(1,1); 
        \fill[black!30] (0,3) rectangle ++(1,1); 
        \fill[black!30] (0,5) rectangle ++(1,1); 
        \fill[black!30] (4,1) rectangle ++(1,2); 
        \fill[black!30] (4,4) rectangle ++(1,2); 
        \fill[black!15] (1,2) rectangle ++(2,1); 
        \fill[black!15] (5,1) rectangle ++(1,1); 
        \fill[white] (2,4) rectangle ++(1,1); 
        \fill[white] (4,0) rectangle ++(1,1); 
        
        \draw[thick] (0,0) rectangle (6,6);
        \foreach \i in {1,...,5} {
          \draw (\i,0) -- (\i,6);
          \draw (0,\i) -- (6,\i);
        }
        
        \foreach \p in {(1.5,5.5),(3.5,2.5),(3.5,4.5),(3.5,5.5),(4.5,3.5)} {
          \node[white] at \p {0};
        }
        \foreach \p in {(0.5,1.5),(1.5,3.5),(2.5,0.5),(2.5,1.5),(2.5,3.5),(3.5,0.5),(5.5,0.5),(5.5,4.5)} {
          \node[white] at \p {1};
        }
        \foreach \p in {(0.5,0.5),(1.5,0.5),(1.5,4.5),(2.5,5.5),(3.5,1.5),(3.5,3.5)} {
          \node[white] at \p {2};
        }
        \foreach \p in {(0.5,2.5),(1.5,1.5)} {
          \node at \p {3};
        }
        \foreach \p in {(0.5,3.5),(0.5,5.5),(4.5,1.5),(4.5,2.5),(4.5,4.5),(4.5,5.5)} {
          \node at \p {4};
        }
        \foreach \p in {(1.5,2.5),(2.5,2.5),(5.5,1.5)} {
          \node at \p {5};
        }
        \foreach \p in {(0.5,4.5),(2.5,4.5),(4.5,0.5),(5.5,2.5),(5.5,3.5),(5.5,5.5)} {
          \node at \p {6};
        }
        
        \draw[orange, line width=3pt, line join=round] (0,0) -- (6,0) -- (6,1) -- (5,1) -- (5,2) -- (4,2) -- (4,3) -- (3,3) -- (3,4) -- (2,4) -- (2,5) -- (1,5) -- (1,6) -- (0,6) -- cycle;
      \end{tikzpicture}
      \caption{}
    \end{subfigure}
    \begin{subfigure}[b]{0.3\textwidth}
      \centering
            \begin{tikzpicture}[scale=0.5] 
        \fill[black!90] (1,2) rectangle (5,6); 
        \fill[black!75] (0,0) rectangle (6,2); 
        \fill[black!75] (1,3) rectangle (3,4); 
        \fill[black!75] (5,4) rectangle (6,5); 
        \fill[black!60] (0,0) rectangle (2,1); 
        \fill[black!60] (3,1) rectangle ++(1,1); 
        \fill[black!60] (3,3) rectangle ++(1,1); 
        \fill[black!60] (1,4) rectangle ++(1,1); 
        \fill[black!60] (2,5) rectangle ++(1,1); 
        \fill[black!45] (0,2) rectangle ++(1,1); 
        \fill[black!45] (1,1) rectangle ++(1,1); 
        \fill[black!30] (0,3) rectangle ++(1,1); 
        \fill[black!30] (0,5) rectangle ++(1,1); 
        \fill[black!30] (4,1) rectangle ++(1,2); 
        \fill[black!30] (4,4) rectangle ++(1,2); 
        \fill[black!15] (1,2) rectangle ++(2,1); 
        \fill[black!15] (5,1) rectangle ++(1,1); 
        \fill[white] (2,4) rectangle ++(1,1); 
        \fill[white] (4,0) rectangle ++(1,1); 
        
        \draw[thick] (0,0) rectangle (6,6);
        \foreach \i in {1,...,5} {
          \draw (\i,0) -- (\i,6);
          \draw (0,\i) -- (6,\i);
        }
        
        \foreach \p in {(1.5,5.5),(3.5,2.5),(3.5,4.5),(3.5,5.5),(4.5,3.5)} {
          \node[white] at \p {0};
        }
        \foreach \p in {(0.5,1.5),(1.5,3.5),(2.5,0.5),(2.5,1.5),(2.5,3.5),(3.5,0.5),(5.5,0.5),(5.5,4.5)} {
          \node[white] at \p {1};
        }
        \foreach \p in {(0.5,0.5),(1.5,0.5),(1.5,4.5),(2.5,5.5),(3.5,1.5),(3.5,3.5)} {
          \node[white] at \p {2};
        }
        \foreach \p in {(0.5,2.5),(1.5,1.5)} {
          \node at \p {3};
        }
        \foreach \p in {(0.5,3.5),(0.5,5.5),(4.5,1.5),(4.5,2.5),(4.5,4.5),(4.5,5.5)} {
          \node at \p {4};
        }
        \foreach \p in {(1.5,2.5),(2.5,2.5),(5.5,1.5)} {
          \node at \p {5};
        }
        \foreach \p in {(0.5,4.5),(2.5,4.5),(4.5,0.5),(5.5,2.5),(5.5,3.5),(5.5,5.5)} {
          \node at \p {6};
        }
        
        \draw[orange, line width=3pt, line join=round] (0,0) -- (4,0) -- (4,1) -- (3,1) -- (3,2) -- (2,2) -- (2,3) -- (1,3) -- (1,4) -- (0,4) -- cycle;
      \end{tikzpicture}
      \caption{}
    \end{subfigure}
    \begin{subfigure}[b]{0.3\textwidth}
      \centering
            \begin{tikzpicture}[scale=0.5] 
        \fill[black!90] (1,2) rectangle (5,6); 
        \fill[black!75] (0,0) rectangle (6,2); 
        \fill[black!75] (1,3) rectangle (3,4); 
        \fill[black!75] (5,4) rectangle (6,5); 
        \fill[black!60] (0,0) rectangle (2,1); 
        \fill[black!60] (3,1) rectangle ++(1,1); 
        \fill[black!60] (3,3) rectangle ++(1,1); 
        \fill[black!60] (1,4) rectangle ++(1,1); 
        \fill[black!60] (2,5) rectangle ++(1,1); 
        \fill[black!45] (0,2) rectangle ++(1,1); 
        \fill[black!45] (1,1) rectangle ++(1,1); 
        \fill[black!30] (0,3) rectangle ++(1,1); 
        \fill[black!30] (0,5) rectangle ++(1,1); 
        \fill[black!30] (4,1) rectangle ++(1,2); 
        \fill[black!30] (4,4) rectangle ++(1,2); 
        \fill[black!15] (1,2) rectangle ++(2,1); 
        \fill[black!15] (5,1) rectangle ++(1,1); 
        \fill[white] (2,4) rectangle ++(1,1); 
        \fill[white] (4,0) rectangle ++(1,1); 
        
        \draw[thick] (0,0) rectangle (6,6);
        \foreach \i in {1,...,5} {
          \draw (\i,0) -- (\i,6);
          \draw (0,\i) -- (6,\i);
        }
        
        \foreach \p in {(1.5,5.5),(3.5,2.5),(3.5,4.5),(3.5,5.5),(4.5,3.5)} {
          \node[white] at \p {0};
        }
        \foreach \p in {(0.5,1.5),(1.5,3.5),(2.5,0.5),(2.5,1.5),(2.5,3.5),(3.5,0.5),(5.5,0.5),(5.5,4.5)} {
          \node[white] at \p {1};
        }
        \foreach \p in {(0.5,0.5),(1.5,0.5),(1.5,4.5),(2.5,5.5),(3.5,1.5),(3.5,3.5)} {
          \node[white] at \p {2};
        }
        \foreach \p in {(0.5,2.5),(1.5,1.5)} {
          \node at \p {3};
        }
        \foreach \p in {(0.5,3.5),(0.5,5.5),(4.5,1.5),(4.5,2.5),(4.5,4.5),(4.5,5.5)} {
          \node at \p {4};
        }
        \foreach \p in {(1.5,2.5),(2.5,2.5),(5.5,1.5)} {
          \node at \p {5};
        }
        \foreach \p in {(0.5,4.5),(2.5,4.5),(4.5,0.5),(5.5,2.5),(5.5,3.5),(5.5,5.5)} {
          \node at \p {6};
        }
        
        \draw[orange, line width=3pt, line join=round] (0,0) -- (4,0) -- (4,3) -- (3,3) -- (3,4) -- (0,4) -- cycle;
      \end{tikzpicture}
      \caption{}
    \end{subfigure}

%% file: euclideanDistExample.tex
  \centering
  \begin{subfigure}[t]{0.23\textwidth}
    \centering
    \begin{tikzpicture}[scale=0.8, baseline={(0,0)}] 
        \fill[orange!40] (0,2) rectangle ++(1,1);
        \fill[orange!40] (2,1) rectangle ++(1,1);
        \draw[thick] (0,0) rectangle (3,3);
        \draw (1,0) rectangle (2,3);
        \draw (0,1) rectangle (3,2);
        
        \node at (0.5,2.5) {$0$};
        \node at (0.2,2.2) {\textsf{\tiny{0}}};
        \node at (1.5,2.5) {$1$};
        \node at (1.2,2.2) {\textsf{\tiny{1}}};
        \node at (2.5,2.5) {$3$};
        \node at (2.2,2.2) {\textsf{\tiny{2}}};
        \node at (0.5,1.5) {$2$};
        \node at (0.2,1.2) {\textsf{\tiny{3}}};
        \node at (1.5,1.5) {$1$};
        \node at (1.2,1.2) {\textsf{\tiny{4}}};
        \node at (2.5,1.5) {$0$};
        \node at (2.2,1.2) {\textsf{\tiny{5}}};
        \node at (0.5,0.5) {$3$};
        \node at (0.2,0.2) {\textsf{\tiny{6}}};
        \node at (1.5,0.5) {$2$};
        \node at (1.2,0.2) {\textsf{\tiny{7}}};
        \node at (2.5,0.5) {$1$};
        \node at (2.2,0.2) {\textsf{\tiny{8}}};
        
        \draw[rounded corners] (0.5,-1.5) rectangle (2.5,-3.1);
        \node at (1.5,-1.9) {$0,0,0$};
        \draw (0.5,-2.3) -- ++(2,0);
        \node at (1.5,-2.7) {$5,5,0$};
        \node[darkgray,rotate=90] at (0.2,-2.3) {\textsf{\footnotesize{QUEUE}}};
        
        \node[green!50!black] at (1.3,-1) {top priority};
        \draw[very thick, green!50!black,->] (2.5,-1) to[out=0,in=0,looseness=2] (2.2,-1.9);
        
        \node[blue,rotate=90,left] at (0.75,-4) {pixel index};
        \draw[very thick, blue,->] (0.75,-4) to[out=90,in=255] (1,-2.9);
        \node[blue,rotate=90,left] at (1.5,-4) {root index};
        \draw[very thick, blue,->] (1.5,-4) to[out=90,in=270] (1.5,-2.9);
        \node[blue,rotate=90,left] at (2.25,-4) {distance to root};
        \draw[very thick, blue,->] (2.25,-4) to[out=90,in=285] (2,-2.9);
        
        \draw[white] (0,-8) -- ++(1,0); 
    \end{tikzpicture}
    \caption{}
  \end{subfigure}
  \begin{subfigure}[t]{0.23\textwidth}
    \centering
    \begin{tikzpicture}[scale=0.8, baseline={(0,0)}]  
        \fill[orange] (0,2) rectangle ++(1,1);
        \fill[orange!40] (0,1) rectangle ++(3,1);
        \fill[orange!40] (1,2) rectangle ++(1,1);
        \draw[thick] (0,0) rectangle (3,3);
        \draw (1,0) rectangle (2,3);
        \draw (0,1) rectangle (3,2);
        
        \node at (0.5,2.5) {$0$};
        \node at (0.2,2.2) {\textsf{\tiny{0}}};
        \node at (1.5,2.5) {$1$};
        \node at (1.2,2.2) {\textsf{\tiny{1}}};
        \node at (2.5,2.5) {$3$};
        \node at (2.2,2.2) {\textsf{\tiny{2}}};
        \node at (0.5,1.5) {$2$};
        \node at (0.2,1.2) {\textsf{\tiny{3}}};
        \node at (1.5,1.5) {$1$};
        \node at (1.2,1.2) {\textsf{\tiny{4}}};
        \node at (2.5,1.5) {$0$};
        \node at (2.2,1.2) {\textsf{\tiny{5}}};
        \node at (0.5,0.5) {$3$};
        \node at (0.2,0.2) {\textsf{\tiny{6}}};
        \node at (1.5,0.5) {$2$};
        \node at (1.2,0.2) {\textsf{\tiny{7}}};
        \node at (2.5,0.5) {$1$};
        \node at (2.2,0.2) {\textsf{\tiny{8}}};
        
        \draw[orange!50!black,very thick,->] (.8,2.5) -- ++(.4,0);
        \draw[orange!50!black,very thick,->] (0.5,2.2) -- ++(0,-0.4);
        \draw[orange!50!black,very thick,->] (.75,2.25) -- ++(.5,-0.5);
        
        \draw[rounded corners] (0.5,-1) rectangle (2.5,-4.2);
        \node at (1.5,-1.4) {$5,5,0$};
        \draw (0.5,-1.8) -- ++(2,0);
        \node at (1.5,-2.2) {$1,0,1$};
        \draw (0.5,-2.6) -- ++(2,0);
        \node at (1.5,-3) {$3,0,1$};
        \draw (0.5,-3.4) -- ++(2,0);
        \node at (1.5,-3.8) {$4,0,\sqrt{2}$};
        \node[darkgray,rotate=90] at (0.2,-1.8) {\textsf{\footnotesize{QUEUE}}};
        
        \draw[white] (0,-8) -- ++(1,0); 
    \end{tikzpicture}
    \caption{}
  \end{subfigure}
  \begin{subfigure}[t]{0.23\textwidth}
    \centering
    \begin{tikzpicture}[scale=0.8, baseline={(0,0)}]    
        \fill[orange!40] (0,1) rectangle ++(2,2);
        \fill[orange!40] (2,2) rectangle ++(1,1);
        \fill[orange!40] (1,0) rectangle ++(2,1);
        \fill[orange] (0,2) rectangle ++(1,1);
        \fill[orange] (2,1) rectangle ++(1,1);
        \draw[thick] (0,0) rectangle (3,3);
        \draw (1,0) rectangle (2,3);
        \draw (0,1) rectangle (3,2);
        
        \node at (0.5,2.5) {$0$};
        \node at (0.2,2.2) {\textsf{\tiny{0}}};
        \node at (1.5,2.5) {$1$};
        \node at (1.2,2.2) {\textsf{\tiny{1}}};
        \node at (2.5,2.5) {$3$};
        \node at (2.2,2.2) {\textsf{\tiny{2}}};
        \node at (0.5,1.5) {$2$};
        \node at (0.2,1.2) {\textsf{\tiny{3}}};
        \node at (1.5,1.5) {$1$};
        \node at (1.2,1.2) {\textsf{\tiny{4}}};
        \node at (2.5,1.5) {$0$};
        \node at (2.2,1.2) {\textsf{\tiny{5}}};
        \node at (0.5,0.5) {$3$};
        \node at (0.2,0.2) {\textsf{\tiny{6}}};
        \node at (1.5,0.5) {$2$};
        \node at (1.2,0.2) {\textsf{\tiny{7}}};
        \node at (2.5,0.5) {$1$};
        \node at (2.2,0.2) {\textsf{\tiny{8}}};
        
        \draw[orange!50!black,very thick,->] (2.2,1.5) -- ++(-.4,0);
        \draw[orange!50!black,very thick,->] (2.25,1.75) -- ++(-0.5,0.5);
        \draw[orange!50!black,very thick,->] (2.5,1.8) -- ++(0,0.4);
        \draw[orange!50!black,very thick,->] (2.25,1.25) -- ++(-0.5,-0.5);
        \draw[orange!50!black,very thick,->] (2.5,1.2) -- ++(0,-0.4);
        
        \draw[rounded corners] (0.5,-1) rectangle (2.5,-7.4);
        \node at (1.5,-1.4) {$1,0,1$};
        \draw (0.5,-1.8) -- ++(2,0);
        \node at (1.5,-2.2) {$3,0,1$};
        \draw (0.5,-2.6) -- ++(2,0);
        \node at (1.5,-3) {$2,5,1$};
        \draw (0.5,-3.4) -- ++(2,0);
        \node at (1.5,-3.8) {$4,5,1$};
        \draw (0.5,-4.2) -- ++(2,0);
        \node at (1.5,-4.6) {$8,5,1$};
        \draw (0.5,-5) -- ++(2,0);
        \node at (1.5,-5.4) {$4,0,\sqrt{2}$};
        \draw (0.5,-5.8) -- ++(2,0);
        \node at (1.5,-6.2) {$1,5,\sqrt{2}$};
        \draw (0.5,-6.6) -- ++(2,0);
        \node at (1.5,-7) {$7,5,\sqrt{2}$};
        \node[darkgray,rotate=90] at (0.2,-1.8) {\textsf{\footnotesize{QUEUE}}};
        
        \draw[white] (0,-8) -- ++(1,0); 
    \end{tikzpicture}
    \caption{}
  \end{subfigure}
  \begin{subfigure}[t]{0.23\textwidth}
    \centering
    \begin{tikzpicture}[scale=0.8, baseline={(0,0)}]  
        \fill[orange!40] (0,1) rectangle (3,3);
        \fill[orange!40] (1,0) rectangle (3,3);
        \fill[orange] (0,2) rectangle ++(2,1);
        \fill[orange] (2,1) rectangle ++(1,1);
        \draw[thick] (0,0) rectangle (3,3);
        \draw (1,0) rectangle (2,3);
        \draw (0,1) rectangle (3,2);
        
        \node at (0.5,2.5) {$0$};
        \node at (0.2,2.2) {\textsf{\tiny{0}}};
        \node at (1.5,2.5) {$1$};
        \node at (1.2,2.2) {\textsf{\tiny{1}}};
        \node at (2.5,2.5) {$3$};
        \node at (2.2,2.2) {\textsf{\tiny{2}}};
        \node at (0.5,1.5) {$2$};
        \node at (0.2,1.2) {\textsf{\tiny{3}}};
        \node at (1.5,1.5) {$1$};
        \node at (1.2,1.2) {\textsf{\tiny{4}}};
        \node at (2.5,1.5) {$0$};
        \node at (2.2,1.2) {\textsf{\tiny{5}}};
        \node at (0.5,0.5) {$3$};
        \node at (0.2,0.2) {\textsf{\tiny{6}}};
        \node at (1.5,0.5) {$2$};
        \node at (1.2,0.2) {\textsf{\tiny{7}}};
        \node at (2.5,0.5) {$1$};
        \node at (2.2,0.2) {\textsf{\tiny{8}}};
        
        \draw[orange!50!black,very thick,->] (1.25,2.25) -- ++(-0.5,-0.5);
        \draw[orange!50!black,very thick,->] (1.5,2.2) -- ++(0,-0.4);
        \draw[orange!50!black,very thick,->] (1.8,2.5) -- ++(0.4,0);
        
        \draw[rounded corners] (0.5,-1) rectangle (2.5,-8);
        \node at (1.5,-1.4) {$3,0,1$};
        \draw (0.5,-1.7) -- ++(2,0);
        \node at (1.5,-2.1) {$2,5,1$};
        \draw (0.5,-2.4) -- ++(2,0);
        \node at (1.5,-2.8) {$4,5,1$};
        \draw (0.5,-3.1) -- ++(2,0);
        \node at (1.5,-3.5) {$8,5,1$};
        \draw (0.5,-3.8) -- ++(2,0);
        \node at (1.5,-4.2) {$4,0,\sqrt{2}$};
        \draw (0.5,-4.5) -- ++(2,0);
        \node at (1.5,-4.9) {$1,5,\sqrt{2}$};
        \draw (0.5,-5.2) -- ++(2,0);
        \node at (1.5,-5.6) {$7,5,\sqrt{2}$};
        \draw (0.5,-5.9) -- ++(2,0);
        \node at (1.5,-6.3) {$3,0,\sqrt{2}$};
        \draw (0.5,-6.6) -- ++(2,0);
        \node at (1.5,-7.0) {$4,0,\sqrt{2}$};
        \draw (0.5,-7.3) -- ++(2,0);
        \node at (1.5,-7.7) {$2,0,2$};
        \node[darkgray,rotate=90] at (0.2,-1.8) {\textsf{\footnotesize{QUEUE}}};
    \end{tikzpicture}
    \caption{}
    \label{testlabel}
  \end{subfigure}